\documentclass[reqno]{amsart}

\usepackage{amsfonts,amsthm,amsmath,amssymb,amscd,amsmath,epsf,latexsym,mathrsfs}
\usepackage{graphicx}
\usepackage{float}
\newcommand{\supp}{\mathop\mathrm{supp}\nolimits}
\newcommand{\spanl}{\mathop\mathrm{span}\nolimits} 
\newcommand{\dist}{\mathop\mathrm{dist}\nolimits}

\newcommand{\ie}{\emph{i.e.}}

\newtheorem{theorem}{Theorem}
\newtheorem{proposition}[theorem]{Proposition}
\newtheorem{lemma}[theorem]{Lemma}

\newtheorem{corollary}{Corollary}
\newtheorem{example}{Example}
\newtheorem{remark}{Remark}

\newtheorem{conjecture}{Conjecture}

\newcommand{\De}{\Delta}
\newcommand{\be}{\beta}
\newcommand{\bC}{\mathbb{C}}
\newcommand{\la}{\lambda}
\newcommand{\La}{\Lambda}
\newcommand{\C}{\mathcal C}
\newcommand{\Si}{\Sigma}
\newcommand{\si}{\sigma}
\newcommand{\ga}{\gamma}
\newcommand{\eps}{\epsilon} 
\newcommand{\N}{\mathbb{N}}
\newcommand{\NN}{\mathbb{N}_{0}}
\newcommand{\R}{{\mathbb{R}}}

\newcommand{\dd}{{\rm d}}

\DeclareMathOperator*{\wlim}{w-lim\ }
\DeclareMathOperator*{\slim}{s-lim\ }

\begin{document}
          \numberwithin{equation}{section}

          \title[Spectral asymptotic of QES-quartic potential]% and Yablonskii-Vorob'ev polynomials]
          {On spectral asymptotic of  quasi-exactly solvable quartic potential}% and Yablonskii-Vorob'ev polynomials}

\author[B. Shapiro]{Boris Shapiro}\thanks{B.Sh. is the corresponding author}
\address{   Department of Mathematics,
            Stockholm University,
            S-10691, Stockholm, Sweden}
\email{shapiro@math.su.se}

\author[M.~Tater]{Milo\v{s} Tater}
\address{Department of Theoretical Physics, Nuclear Physics Institute, 
Academy of Sciences, 250\,68 \v{R}e\v{z} near Prague, Czech
Republic}
\email{tater@ujf.cas.cz}

\date{\today}
%,asymptotic root distribution}

\begin{abstract} Motivated by the earlier results of \cite {MuTa} and \cite {ShTaTa}, we discuss below the  asymptotic of the solvable part  of the spectrum for the quasi-exactly solvable quartic oscillator.  % introduced by C.~M.~Bender and S.~Boettcher in \cite{BeBo}. 
 In particular,  we formulate a conjecture on the  coincidence of the asymptotic shape  of the level crossings of the latter oscillator with the asymptotic shape of  zeros of the Yablonskii-Vorob'ev polynomials.   Further we present a numerical study of the spectral monodromy for the oscillator in question.%, see Conjecture~\ref{conj:scaling}.

\end{abstract}
\maketitle
%\footnote{B.~Shapiro is the corresponding author }

\section{Introduction}

\medskip
A quasi-exactly solvable quartic oscillator was introduced by C.~M.~Bender and S.~Boettcher  in \cite{BeBo} and (in its restricted form) is a Schr\"odinger-type eigenvalue problem of the form
\begin{equation}\label{eq:QESQU}
L_J(w)=w^{\prime\prime}-(x^4/4-ax^2/2-Jx)w=\la w
\end{equation}
with the boundary conditions $w(t)\to 0$ and $w(te^{2\pi i/3})\to 0$ as $t\to +\infty,$ where $a\in \bC$ and $J$ are parameters of the spectral problem.  %\footnote{This paper written by a mathematician and a physicist contains   three  types of results: rigorous mathematical statements, statements proven of the physics level of rigor, i.e., under certain assymptions of convergence and, finally, conjectures obtained after substantial numerical experimentation. It best fits under the category of experimental mathematics and illustrates the sagacious saying of the late Vladimir Arnold:  
%\centerline{"Mathematics is a branch of physics in which experiments are cheap.``}}  
With these boundary conditions, real $a$ and $J$, \eqref{eq:QESQU} is not a Hermitian but is a $PT$-symmetric operator,  see \cite{Shin}, \cite{MuTa} and \cite {EG}. If $J=n+1$ is a positive integer, then $L_{n+1}(w)$ maps the linear space of quasi-exactly solutions of the form $\{pe^{-x^3/6 + ax/2}: \deg p \le n\}$  to itself where $p$ belongs to the linear space of  univariate polynomials of degree at most $n$. The restriction of the operator $L_{n+1}(w)$ to the latter linear space is a finite-dimensional linear operator whose spectrum and eigenfunctions can be found explicitly using methods of linear algebra. This part of the spectrum and eigenfunctions of \eqref{eq:QESQU} is usually  referred to as {\em solvable}. (Observe that the operator $L_J(w)$ given by \eqref{eq:QESQU} has a negative spectrum while physicists usually  prefer to work with $-L_J(w)$ whose spectrum is positive.) 

\medskip
One can easily  show that  polynomial factors $p$ in the quasi-exactly solutions $w=pe^h$ of \eqref{eq:QESQU} coincide with the  polynomial solutions of the degenerate Heun equation 
\begin{equation}\label{eq:degHeun}
y''-(x^2-a)y'+(\alpha x-\la)y=0,
\end{equation}
where $a\in \bC$ has the same meaning as above and $(\alpha,\la)$ are the spectral variables.   Obviously,  if equation \eqref{eq:degHeun} has a polynomial solution 
of degree $n,$ then $\alpha=n$. Furthermore, to get a polynomial solution of \eqref{eq:degHeun} of degree $n$,  the remaining spectral variable 
 $\la$ should be  chosen as an eigenvalue of the operator 
 \begin{equation}\label{eq:oper}
 T_n(y)=y''-(x^2-a)y'+n xy
 \end{equation}
 acting on the linear space of polynomials of degree at most $n.$ 

\smallskip 
 In the standard monomial basis $\{1,x,x^2,\dots,x^n\}$ of the latter linear space, the action of the operator $T_n$ is given  by  the 
 $4$-diagonal $(n+1)\times(n+1)$-matrix of the form 

\begin{center}
\hspace{-2cm}
\begin{equation}\label{matrix}
M_n^{(a)} := \left(
\begin{matrix}
0 & a & 2 & 0 & 0 & \cdots & 0 \\
n & 0 & 2a & 6 & 0 & \cdots & 0 \\
0 & n-1 & 0 & 3a & 12 & \cdots & 0 \\
\vdots & \vdots & \ddots & \ddots & \ddots & \ddots & \vdots \\
0 & 0 & \cdots & 3 & 0 & (n-1)a & n(n-1) \\
0 & 0 & \cdots & 0 & 2 & 0 & na \\
0 & 0 & \cdots & 0 & 0 & 1 & 0
\end{matrix}
\right).
\end{equation}
\end{center}
%see Lemma~\ref{linalg} below. %(The spectral parameter is called $b$). 

\medskip
We will call the bivariate characteristic polynomial $Sp_n(a,\la):=\det(M_n^{(a)}-\la I)$ the {\em $n$-th spectral polynomial of \eqref{eq:degHeun}}. The  degree of  $Sp_n(a,\la)$ equals $n+1$ which is also its degree with respect to the variable $\la$. (The  maximal degree of the variable $a$ in  $Sp_n(a,\la)$ equals $\left[\frac{n+1}{2}\right]$).      
Additionally observe that for $a\in \R$,  $Sp_n(a,\la)$ is a real polynomial in $\la$ and therefore the spectrum of $M_n^{(a)}$ is symmetric with respect to the real axis. %Numerics shows that changing the %parameter $a$ to its complex-conjugate $\bar a$  changes the spectrum to its complex-conjugate which is apparently trivial to prove, but we have not done it yet. 

\medskip
The main goal of this paper is to study the asymptotic of the spectrum of $M_n^{(a_n)}$  in two different regimes.  In the first regime we require that $\lim_{n\to\infty}\frac{a_n}{n^{2/3}}=0$  while  in the second regime we require that  $\lim_{n\to\infty}\frac{a_n}{n^{2/3}}=A\neq 0$. %In other words, we consider solutions of $Sp_n(a_n,\la)=0$ with respect to  $\la$ in the above two situations. 

\medskip
\noindent
{\em Remark.} In case of a non-degenerate Heun equation detailed study of similar asymptotic was carried out  in \cite {ShTa} and \cite {ShTaTa}. 

\subsection{Main results} \hfill \\ 

\begin{theorem}\label{th:main}
\rm{(i)} If  $\lim_{n\to\infty}\frac{a_n}{n^{2/3}}=0$, the maximal absolute value $r_n(a_n)$ of the eigenvalues of $M_n^{(a_n)}$   %i.e., the maximal absolute value among the roots of $Sp_n(a_n,\la),$  
   grows as $\frac{3}{4} n^{4/3}$, i.e.,  $r_n(a_n)=\frac{3}{4} n^{4/3}(1+o(1)).$

\medskip
\noindent 
\rm{(ii)} If  $\lim_{n\to\infty}\frac{a_n}{n^{2/3}}=0$,   the sequence $\{\mu_n^{(a_n)}\}$ of eigenvalue-counting measures  %for the sequence $\{Sp_n(a_n,\be n^{4/3})\}$, i.e., 
 for the spectra of the sequence of matrices $\{\frac{1}{n^{4/3}} M_n^{(a_n)}\}$  weakly  converges   to the measure $\nu_0$ supported on the union of three straight intervals connecting the origin with three cubic roots of  $\frac{27}{64}$, see Figure~\ref{fig1}.
\end{theorem}

More information about $\nu_0$ can be found in \S~2.

\begin{figure}[H]

\begin{center}
\includegraphics[width=0.4\textwidth]{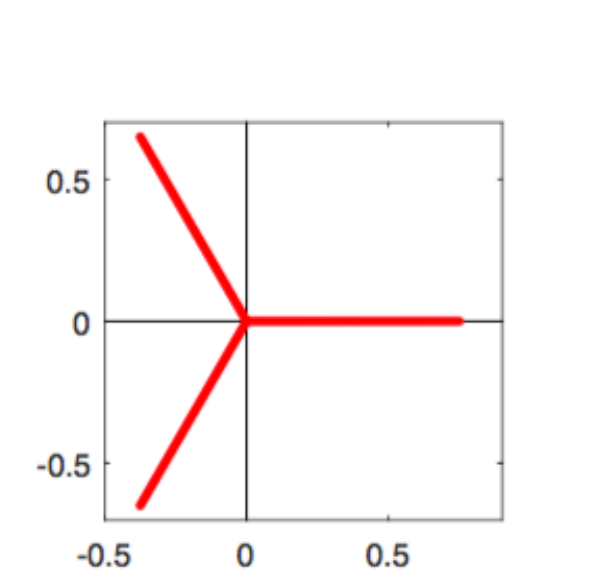}
\end{center}

%\vskip 1cm

\caption{ Distributions of the eigenvalues of $M_n^{(0)}$ scaled by $n^{4/3}$ for n=200. }
\label{fig1}
\end{figure}

%\begin{figure}

%\begin{center}
%\includegraphics[clip, trim=3.5cm 6cm 4cm 10cm, width=0.6 \textwidth]{a=0-200.pdf}
%\end{center}

%\vskip 1cm

%\caption{ Distributions of the eigenvalues of $M_n^{(0)}$ scaled by $n^{4/3}$ for n=200. }
%\label{fig1}
%\end{figure}

%\begin{figure}

%\begin{center}
%\includegraphics[scale=0.4]{a=70I-600-500.pdf} \quad \quad \includegraphics %[scale=0.4]%{a=70I-600-700.pdf}  \quad \quad \includegraphics [scale=0.4]{a=70I-800-700.pdf}
%\end{center}

%\vskip 1cm

%\caption{ Distributions of the eigenvalues of $M_n^{(70I)}$ scaled by $n^{4/3}$. The leftmost picture shows %$n=500$ (red) and $n=600$ (blue); the central picture shows $n=600$ (blue) and $n=700$ (red); the %rightmost picture shows $n=700$ (red) and $n=800$ (blue).}
%\label{fig1}
%\end{figure}

\medskip
Our next result is obtained on the physics level of rigor, i.e., modulo three convergence assumptions which are explicated  in the proof of Proposition~\ref{quart}, see \S~\ref{sec:phys}. 
%Consider the sequence $\{a n^{2/3}\}_{n=0}^\infty,$ where $a\in \bC$ is  fixed. Take the sequence $\{Sp_n(an^{2/3}, \be %n^{4/3})\}$ of the corresponding  spectral polynomials  (where we use the scaling $\la=\be n^{4/3}$) and
 To formulate it,  take the family  of equations 
\begin{equation}\label{quadreq}
\C^2-(x^2-A)\C+(x-\be)=0
\end{equation}
which we  consider as   quadratic equations in the variable $\C$ depending on the main variable $x$ and additional parameters $A$ and $\be$.

\medskip
\begin{proposition}\label{quart}   {\rm  If  $\lim_{n\to\infty}\frac{a_n}{n^{2/3}}=A\neq 0$, then (under three additional convergence assumptions  given in the proof) the sequence $\{\mu_n^{(a_n)}\}$ of eigenvalue-counting measures for the sequence  of matrices $\{\frac{1}{n^{4/3}}M_n^{(a_n)}\}$  weakly converges  to a special compactly supported probability measure $\nu_A$.

 In particular, 
the support of $\nu_A$ consists of all values of the spectral parameter $\be$ for which there exists a compactly supported  in the $x$-plane probability measure $\kappa$  whose Cauchy transform $\C_\kappa(x)$ satisfies 
\eqref{quadreq} in this plane almost everywhere.  } % the differential \eqref{quadquart} has two critical horizontal trajectories connecting both  pairs of its zeros.  
%CHECK!!! Maybe quadratic equation admitting a positive measure?
\end{proposition}

\begin{remark}{\rm 
Numerical experiments and some theoretical considerations indicate  that if $\lim_{n\to\infty}\frac{a_n}{n^{2/3}}$ does not exist then there is no non-trivial limiting behavior of the spectrum. We think that for $\lim_{n\to \infty}\frac{|a_n|}{n^{2/3}}=+\infty$, one can rigorously prove the latter observation. However we are not trying to pursue this aspect in the present paper.}
\end{remark} 

\begin{remark}{\rm 
To make the paper self-contained, let us recall that for a complex-valued measure $\mu$ compactly supported in $\bC$, its logarithmic potential is defined as 
 $$\mathfrak u_\mu(x):=\int_\bC \ln|x-\xi|d\mu(\xi)$$ and   its Cauchy transform is defined as
$$\C_\mu(x):=\int_\bC \frac{d\mu(\xi)}{x-\xi}=\frac{\mathfrak \partial \mathfrak u_\mu(x)}{\partial x}.$$}
\end{remark} 

\begin{remark}{\rm 
The existence of a signed (real, but not necessarily positive) measure $\mu$  whose Cauchy transform $\C_\mu(x)$ satisfies 
\eqref{quadreq} almost everywhere in $\bC$ is closely related to the properties of the family of quadratic differentials  
\begin{equation}\label{quadquart}
\Psi_{A,\be}=-((x^2-A)^2-4(x-\be))dx^2
\end{equation}
depending on parameters  $A$ and $\be$.  We will present some information about this connection in \S~\ref{sec:phys}. (For an accessible  introduction to quadratic differentials, see e.g. \cite {Str}).}
\end{remark} 

%\smallskip
%Proposition~\ref{quart} reminds us of the main result of \cite{BBSh2} where we considered the so-called homogenized spectral problem. In this case  all terms of the Heun differential operator in question are important for the root asymptotics of  sequences of eigenpolynomials.  Finally, our numerical experiments strongly support the validity of Proposition~\ref{quart}.} 
%\end{remark}

\subsection{Asymptotic  distribution of  the branching points  of $Sp_n(a,\la)$ and Yablonskii-Vorob'ev polynomials} \hfill \\

\medskip
To finish the introduction,  let us formulate a number of conjectures supported by extensive numerical experiments. For a given positive integer $n$,  denote by $\Si_n$ the set of all branching points of the projection of the algebraic curve $\Gamma_n(a): \{Sp_n(a,\la)=0\}$ to the $a$-axis. In other words, $\Si_n$ is the set of all values of the complex parameter $a$ for which the matrix 
$M_n^{(a)}$ has a multiple eigenvalue, i.e.,  $Sp_n(a,\la)$ has a multiple root. In physics literature such points are called \emph{level crossings}.  Obviously, one can describe   $\Si_n$ as  the zero locus of the univariate discriminant polynomial $D_n(a)$ which is the resultant of $Sp_n(a,\la)$ and $\frac{\partial Sp_n(a,\la)}{\partial\la}$ with respect to $\la$.   
One can show that the degree of $D_n(a)$ equals $\binom{n+1}{2}$. 

\medskip
Further recall that Yablonskii-Vorob'ev polynomials $\{Q_n\}$ are defined as follows, see e.g. \cite{Vo, Ya}. Set $Q_0=1,\; Q_1=t$, and  for $n\ge 1$, define  
$$Q_{n+1}=\frac {t \cdot Q_n^2-4(Q_n \cdot Q^{\prime\prime}_n-(Q^\prime_n)^2)}{Q_{n-1}}.$$
Although the latter expression a priori determines  a rational function, $Q_n$ is in fact a polynomial of degree $\binom {n+1}{2}$, see e.g. \cite{Tan}.  The importance of Yablonskii-Vorob'ev polynomials is explained by the fact that all rational solutions of the second Painlev\'e equation
$$u_{tt}=tu+2u^3 +\alpha, \alpha \in \bC,$$
are presented in the form
$$u(t)=u(t;n)=\frac{d}{dt}\left\{\ln\left[\frac{Q_{n-1}(t)}{Q_n(t)}\right]\right\},\; u(t,0)=0,\; u(t;-n):=-u(t;n).$$
Denote by $\mathcal Z_n$ the zero locus of  $Q_n$. (Good exposition of the properties of $Q_n$ together with several pictures of $\mathcal Z_n$ can be found in \cite{ClMa}). 
 Our conjectures below reveal an unexpected connection of $\Si_n$ and $\mathcal Z_n$.

\begin{remark}
{\rm One can show that the maximal absolute value of points  in $\mathcal Z_n$ grows as $ \frac{3}{\root 3 \of 2}  n^{2/3}$, see \cite{BM} and  \S~\ref{sec:YV}.  Similarly, the maximal absolute value of points  in $\Si_n$ grows as $ \frac{3}{\root 3 \of 4}  n^{2/3}$, see  Lemma~\ref{apositive} and Corollary~\ref{cor:si} below. } \end{remark}

\begin{conjecture}\label{con:lat} Given a positive integer $\ell$, define $\Upsilon_\ell:=\Si_{3\ell}\cup\Si_{3\ell+1}\cup \Si_{3\ell+2}$. For $R>0$, let $K_R$ be the square with the side $2R$ centered at the origin.  Then for any fixed $R>0$ and $\ell\to \infty$, the points in $\Upsilon_\ell$ converge  inside $K_R$  
to the nodes of a certain fixed hexagonal lattice.
\end{conjecture} 

\begin{figure}[H] 

\begin{center}

\includegraphics[scale=0.5]{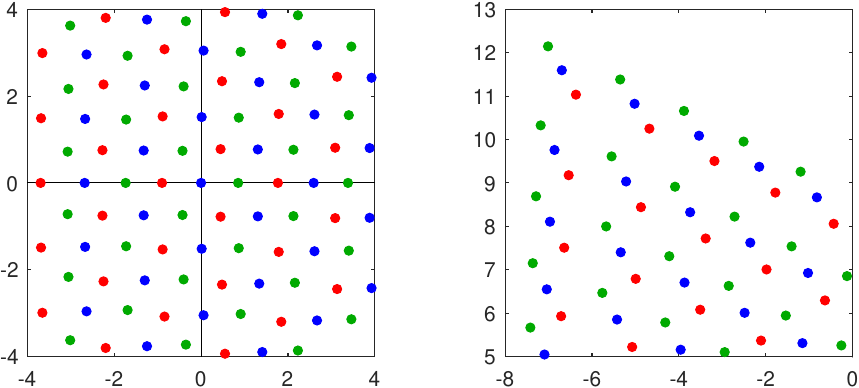}

\end{center}

\caption {The points of $\Upsilon_5$ inside the square $K_4$ (left) and close to the upper left corner (right). The points of $\Si_{15}$ are shown in red, $\Si_{16}$ in blue, and $\Si_{17}$ in green.}

\label{lattice}
\end{figure}

\medskip
\begin{conjecture}\label{conj:Yab} Set $\widetilde \Si_n=\root {3} \of {2}\cdot \Si_n$, i.e., multiple every point in $\Si_n$ by $\root {3} \of {2}$. Then every point in $\widetilde \Si_n$ lies very close to the unique point in $- \mathcal Z_n$ and vice versa, see  Fig.~\ref{BBigtriangle}. Fixing $n$, define $d(n):=max_{p\in \widetilde \Si_n} \min_{q\in -\mathcal Z_n} d(p,q)$, i.e., $d(n)$ is the maximal distance between  points in  $\widetilde \Si_n$ and their respective nearest points in $- \mathcal Z_n$. 

 The sequence $\{d(n)\}$ is very slowly growing with $n$, see Example~\ref{ex:ex} below. It might  even have a limit when $n\to \infty$. Moreover for any fixed $R>0$, the sequence $d_R(n)$ converges to $0$ where $d_R(n)$ is a similar maximin of the pairwise distances taken over all points in  $\widetilde \Si_n$ and $-\mathcal Z_n$ which lie inside the square $K_R$. 
 \end{conjecture}
 
 \begin{example}\label{ex:ex} {\rm Numerical experiments show that for   $n=10,  15,   20,  25, 30,  35,  40$, the corresponding values of $d(n)$ are  approximately $0.03016, 0.04160,0.051156,   0.05837, $ $ 0.06378, 0.06863, 0.07272$ respectively.} 
\end{example}

\begin{figure}
\begin{center}
\includegraphics[scale=0.45]{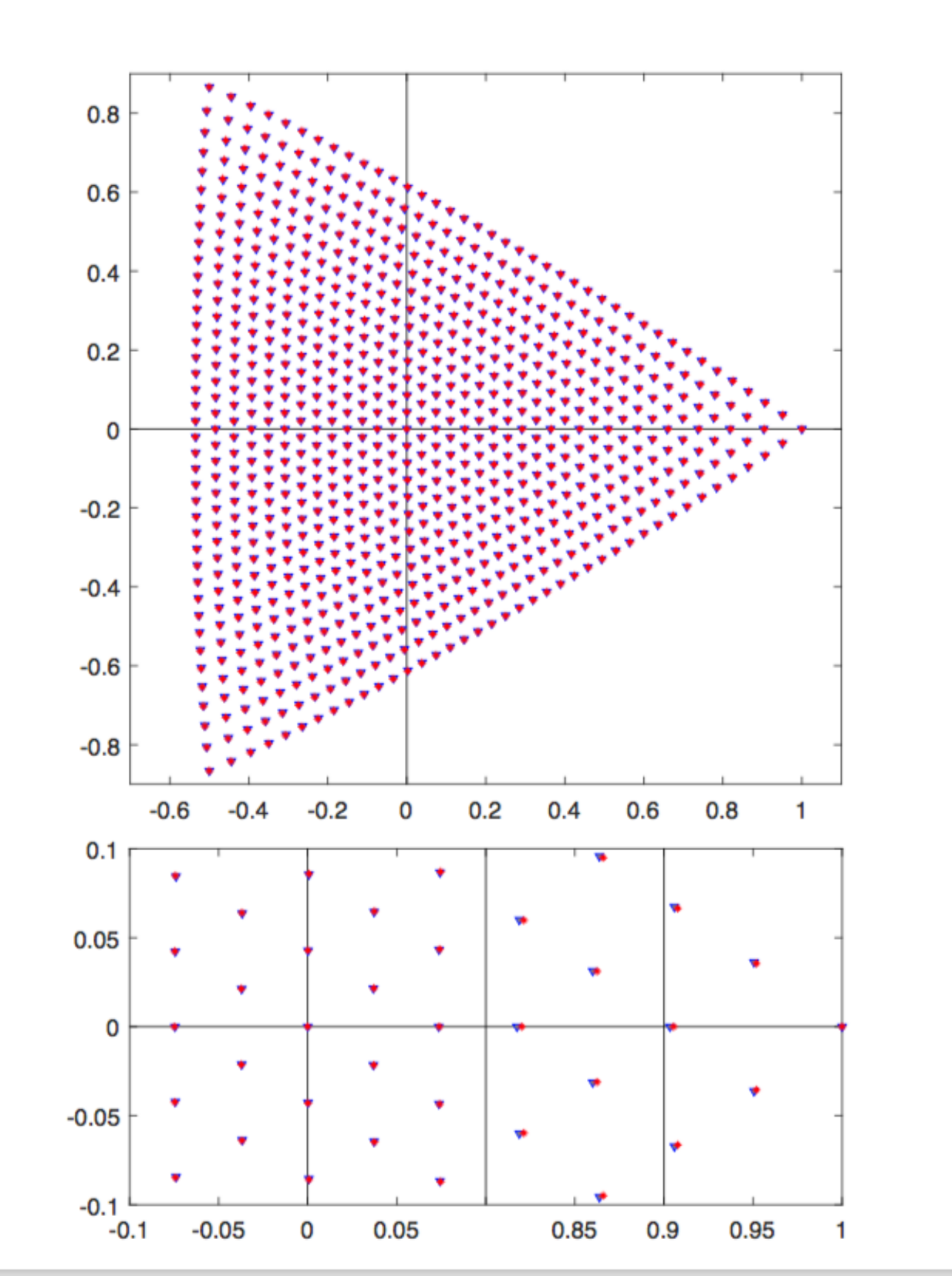}  %width=5.0cm, 
\end{center}

\caption{The  set $\Si_{40}$ of the branching points for $Sp_{40}(a,\la)$ together with the zero locus $-{\mathcal Z}_{40}$  after scaling sending the three corners to the cubic roots of $1$. The blue dots are the scaled roots of $-{\mathcal Z}_{40}$ while the red ones are the scaled branching points of $Sp_{40}(a,\la)$. The bottom picture shows a fragment of the top one with substantial magnification.  (Observe the surprising closeness; the distinction is hardly visible by a naked eye!)}\label{BBigtriangle}
\end{figure}

%\begin{conjecture}\label{conj:scale}
%After division by  $\frac{3}{\root {3} \of {4}} n^{2/3}$ and $ \left(\frac{9}{2}\right)^{\frac{2}{3}}  n^{2/3}$ respectively, the sequence  $\{\Si_n\}$ tends IN WHAT SENSE ??? to the sequence $\{\widehat{\mathcal Z}_n\}$, where $\widehat{\mathcal Z}_n$ is the zeros locus of $Y_n(-t)$. In particular, the limiting triangular domains  %$\mathfrak F$ 
% covered by $\{\Si_n\}$ and $\{\widehat{\mathcal Z}_n\}$  when $n\to \infty$ coincide after the latter scaling. We denote the latter domain by $\mathfrak F$. 
 % \end{conjecture} 

% \begin{conjecture}\label{conj:Yab2}
% \medskip
%\noindent
% {\rm (ii)} The interior  of $\mathfrak F$ consists of all values of   $A,$ for which the support of the above measure $\nu_A$ has a tripod form, i.e., it has three legs emanating from the same point.    The complement of the closure of $\mathfrak F$ consists of all values of   $A,$ for which the support of measure $\nu_A$ is a simple smooth arc (i.e. no self-intersections), see Figure~\ref{figA}. 
%\end{conjecture}

For illustration of Conjecture~\ref{conj:Yab}   see Figure~\ref{BBigtriangle}. We hope that considerations similar to that in \cite{Ma, MaBe} can help to settle it.

\begin{conjecture}\label{conj:scaling}
When $n\to\infty$,  the points in the sequences $ \frac{\root {3} \of {4}}{3 n^{2/3}} \{\Si_{n}\}$ and $\{ -\frac{\root {3} \of {2}}{3 n^{2/3}} {\mathcal Z}_{n}\}$  asymptotically fill the same   curvilinear triangular shape $\mathfrak F$,  see Fig~\ref{BBigtriangle}. 

 The interior of  $\mathfrak F$ consists  of all values  $a\in \bC$ for which the  support of measure $\nu_a$ introduced in Proposition~\ref{quart} is a tripod, i.e., consists of three smooth segments with a common point, see Figure~\ref{figA} (left). 
 
 The complement of $\mathfrak F$ consists of all values  $a\in \bC$ for which the  support of $\nu_a$ is a single smooth segment, see Figure~\ref{figA} (down). The boundary of $\mathfrak F$ consists of those $a\in \bC$ for which the support of $\nu_a$ is a single curve with a singular point, i.e., $a$ belongs to the boundary between the domain where this support is a tripod  and the domain where the support is a smooth single curve, see Figure~\ref{figA} (right). 
\end{conjecture}

Conjecture~\ref{conj:scaling} complements Conjecture~\ref{conj:Yab}. For the sequence $\{{\mathcal Z}_{n}\}$ parts of the latter conjecture have been settled in \cite{BM}, see also \cite{BeBot}.

\begin{remark} Data sharing not applicable to this article as no datasets were generated or analysed during the current study.
\end{remark}

\medskip
\noindent
{\em Acknowledgements.} The authors are sincerely grateful to M.~Duits, A.~Kuijlaars, and A.~Gabrielov for  discussions. The first author wants to thank F.~{\v S}tampach for his help with the complex version of the main result of \cite{KvA} and  Nuclear Physics Institute at \v{R}e\v{z} near Prague for the hospitality. The first author is partially supported by International Laboratory of Cluster Geometry NRU HSE, RF Government grant, ag. № 075-15-2021-608 from 08.06.2021.  The second author was supported by the Czech Science Foundation (GA\v CR) within the project 21-07129S. We are very thankful to our anonymous referees for their careful reading of the original submission and a large number of useful suggestions.   %We are greatly indebted to the anonymous referee for the careful reading of the manuscript and many illuminating hints which helped us both to improve the quality of the exposition and, more importantly, to find extra connections of the above results  with that of \cite{BM}. 

\section {Case  $\lim_{n\to\infty}\frac{a_n}{n^{2/3}}=0$. Proof of Theorem~\ref{th:main}}\label{sec:2}

Let us start our considerations of the spectral asymptotic  of sequences $\{M_n^{(a_n)}\}$ with the simplest case $a_n=0,\; n=1,2,3,\dots$.  %The next numeric observations are  probably easy to prove rigorously, see Figure~\ref{fig1}. 

\begin{proposition}\label{conj:2}{\rm

\medskip
 \noindent 
 \rm{(i)} 
 The
sequence $Sp_{n}(0,\la)$ splits into the following three subsequences: 

\begin{enumerate}
\item{for $n+1=3\ell$, the polynomial $Sp_{n}(0,\la)$ contains only the 
powers of $\la$ divisible by $3$; %introducing $\xi=\la^3,$ 
i.e., we have
$Sp_{n}(0,\la)=u^{(\ell)}(\la^3)$;} 
\item{for $n+1=3\ell+1$, one has $Sp_{n}(0,\la)=\la v^{(\ell)}(\la^3)$;}
\item{for $n+1=3\ell+2$, one has 
$Sp_{n}(0,\la)=\la^2w^{(\ell)}(\la^3)$,}
\end{enumerate}
where $u^{(\ell)}(\xi), v^{(\ell)}(\xi)$, and $w^{(\ell)}(\xi)$ are polynomials of degree $\ell$ in the variable $\xi$.   

\medskip
 \noindent 
 \rm{(ii)}  All three polynomials $u^{(\ell)}(\xi), v^{(\ell)}(\xi),
w^{(\ell)}(\xi)$ 
have simple and  negative  roots. } 
\end{proposition}

\begin{corollary} {\rm (a) The spectrum of $M_n^{(0)},$ i.e., the zero locus of $Sp_n(0,\la)$, is invariant under 
 the rotation by the angle $2\pi/3$ around the origin. 
 
 \noindent
 (b)  For any positive integer $n$,  the spectrum of $M_n^{(0)}$ is  located on the union of the three rays through the origin as illustrated in Figure~\ref{fig1}.} %, i.e., it has a $\mathbb Z_3$-symmetry.  
\end{corollary}

To prove Proposition~\ref{conj:2} following  the ideas of \cite {KvA}, we need to introduce a double indexed polynomial sequence containing our original sequence of characteristic polynomials $\{Sp_n(0,\la)\}$. The most natural way to do this is to consider the principal minors of $M_n^{(0)}-\la I$ given by: 
\begin{center}
\hspace{-2cm}
\begin{equation}\label{Nmatrix}
M_n^{(0)} -\la I:= \left(
\begin{matrix}
-\la & 0 & 2 & 0 & 0 & \cdots & 0 \\
n & -\la & 0 & 6 & 0 & \cdots & 0 \\
0 & n-1 & -\la & 0 & 12 & \cdots & 0 \\
\vdots & \vdots & \ddots & \ddots & \ddots & \ddots & \vdots \\
0 & 0 & \cdots & 3 & -\la & 0 & n(n-1) \\
0 & 0 & \cdots & 0 & 2 & -\la & 0 \\
0 & 0 & \cdots & 0 & 0 & 1 & -\la
\end{matrix}
\right).
\end{equation}
\end{center}

\medskip
Namely, denote by $\De_n^{(k)}(\la)$ the $k$-th principal minor of \eqref{Nmatrix}. Then $$Sp_n(0,\la)=\De_n^{(n+1)}(\la)=\det(M_n^{(0)} -\la I) .$$ Since the latter matrix is $4$-diagonal with one subdiagonal and two superdiagonals, then, by a general result of \cite{PeZa}, its principal minors satisfy  a $4$-term linear recurrence relation. Simple explicit calculation gives %a $4$-term relation: 
\begin{equation}\label{recrel}
\De_n^{(k)}(\la)=-\la\De_n^{(k-1)}(\la)+(n-k+2)(n-k+3)(k-1)(k-2)\De_n^{(k-3)}(\la),
\end{equation}
where $n$ is fixed and $k$ runs from $1$ to $n+1,$ with the standard initial conditions 
$$\De_n^{(-2)}(\la)=\De_n^{(-1)}(\la)=0, \quad \De_n^{(0)}(\la)=1.$$
 Notice that recurrence~\eqref{recrel}  has variable coefficients depending both on $k$ and $n$. 

\medskip
Another form of \eqref{recrel} which is a bit easier to study is as follows. To simplify our manipulations with the signs, let us instead of $M_n^{(0)}-\la I$  consider  the principal minors of $M_n^{(0)}+\la I$. Introducing $\nabla_n^{(k)}(\la):=\De_n^{(k)}(-\la)$, we get the  recurrence:  
\begin{equation}\label{recrelmod1}
\nabla_n^{(k)}(\la)=\la\nabla_n^{(k-1)}(\la)+(n-k+2)(n-k+3)(k-1)(k-2)\nabla_n^{(k-3)}(\la)
\end{equation}
with the initial conditions 
$$\nabla_n^{(-2)}(\la)=\nabla_n^{(-1)}(\la)=0, \quad \nabla_n^{(0)}(\la)=1.$$

\medskip
\begin{lemma}\label{lm:1} Set $\nabla_n^{(3j)}(\la)=U_n^{(j)}(\xi)$, 
$\nabla_n^{(3j+1)}(\la)=\la V_n^{(j)}(\xi)$, $\nabla_n^{(3j+2)}(\la)=\la^2W_n^{(j)}(\xi)$,  where $\xi=\la^3$ and  $U_n^{(j)}, V_n^{(j)}, W_n^{(j)}$ are monic polynomials of degree $j$.

\medskip
Using $U_n^{(j)}, V_n^{(j)}, W_n^{(j)}$,  recurrence~\eqref{recrelmod1} can be rewritten as the system: 
\begin{equation}\label{eq:rec3}
\begin{cases} 
U_n^{(j)}(\xi)=\xi  W_n^{(j-1)}(\xi)+(n-3j+2)(n-3j+3)(3j-2)(3j-1) U_n^{(j-1)}(\xi)\\
V_n^{(j)}(\xi)= U_n^{(j)}(\xi)+(n-3j+1)(n-3j+2)(3j-1)3j\,  V_n^{(j-1)}(\xi)\\
W_n^{(j)}(\xi)= V_n^{(j)}(\xi)+(n-3j)(n-3j+1)3j(3j+1)  W_n^{(j-1)}(\xi)\
\end{cases}
\end{equation}
with the initial conditions $U_n^{(0)}(\xi)=V_n^{(0)}(\xi)=W_n^{(0)}(\xi)=1$ where for any fixed $n$, $j$ runs from $1$ to $[n/3]$. 
\end{lemma}

\begin{proof}
Simple algebra. 
\end{proof}

\begin{proof}[Proof of Proposition~\ref{conj:2}] Item (i)  follows immediately from Lemma~\ref{lm:1}. 
While proving item (ii) of Proposition~\ref{conj:2}, we  will use  notation of Lemma 4. Notice that polynomials $U_n^{(\ell)}$, $V_n^{(\ell)}$, $W_n^{(\ell)}$ coincide up to the change of sign of the variable $\xi$ with  polynomials $u^{(\ell)},$ $v^{(\ell)}$ and $w^{(\ell)}$ respectively.

We need to  show that for any positive integer $n$, each polynomial in the recurrence~\eqref{eq:rec3} has positive coefficients and  real negative roots. Moreover  the roots of any two consecutive polynomials in each of the three sequences  $\{U_n^{(j)}\}_{j=0}^{[n/3]}$, $\{V_n^{(j)}\}_{j=0}^{[n/3]}$, $\{W_n^{(j)}\}_{j=0}^{[n/3]}$ are strictly interlacing and are therefore  simple. 

  Positivity of the coefficients of polynomials is straightforward from the positivity of coefficients in \eqref{eq:rec3} and the initial conditions. Let us settle the negativity of all roots by using induction on $j$. 
 For   $j=1$ and $n\ge 3$, it is  trivial  to check that the negative root of the $U_n^{(1)}$ is larger than that of $V_n^{(1)}$  which is larger than that of $W_n^{(1)}$. 
 
 Let us consider   the case $j=2$.   Note that $U_n^{(2)}=\xi  W_n^{(1)}+20(n-4)(n-3)  U_n^{(1)} $. 
Elementary calculations give  $\xi  W_n^{(1)}=\xi^{2}+(84-8n+20n^{2})\xi$ and $U_n^{(1)}=\xi-2n+2n^{2}$, which implies  that the two roots of $\xi  W_n^{(1)}$  are non-positive  and the only root of $U_n^{(1)}$ is located strictly between them for $n\geq3$.  
Hence, then 
using  Lemma 1.10 of \cite{F} or conducting elementary calculations, we can conclude that $U_n^{(2)}$ has simple negative roots for $n\ge 3$. Using  similar elementary calculations we obtain similar results for the pair $U_n^{(2)},V_n^{(1)}$ implying that $V_n^{(2)}$ has simple negative roots and for the pair $V_n^{(2)},W_n^{(1)}$  implying  that $W_n^{(2)}$ has simple negative roots. 

\smallskip
Assume now that our hypothesis holds for a given positive integer $j$. Note that, in each of the three equations in \eqref{eq:rec3}, the degree of the first polynomial in the right-hand side is bigger than the degree of the second polynomial by one. 
This indicates that the largest root belongs to the polynomial with the larger degree. Furthermore by Corollary 1.30 of \cite {F}, we can derive the following results:
\begin{equation}
 \xi W_{j-1}\leftarrow U_{j} \leftarrow U_{j-1}
 \end{equation}
 \begin{equation}
 U_{j}\leftarrow V_{j} \leftarrow V_{j-1}  
 \end{equation}
\begin{equation}
  V_{j}\leftarrow W_{j} \leftarrow W_{j-1}. 
\end{equation}
Here the arrow $"\leftarrow ``$ indicates that the corresponding pair of polynomials have simple interlacing roots with the largest root belonging to the polynomial at which  the arrow points.
\begin{enumerate}
 \item [(a)] Consider the recurrence 
$$U_n^{(j+1)}=\xi  W_n^{(j)}+(n-3j-1)(n-3j)(3j+1)(3j+2) U_n^{(j)}.$$ 
We can rewrite $\xi  W_n^{(j)}$ as $\xi  V_n^{(j)}+(n-3j)(n-3j+1)3j(3j+1)  \xi W_n^{(j-1)}.$
 Observe that  $\xi W_n^{(j-1)}\leftarrow U_n^{(j)}$ by induction hypothesis and Corollary 1.30 in \cite{F}. Using (2.2), we can conclude that $\xi V_n^{(j)}\leftarrow U_n^{(j)}$ since all roots of $U_n^{(j)}$ are negative and simple. Hence $\xi W_n^{(j)}\leftarrow U_n^{(j)}$, by Lemma 1.31 in \cite{F}.
Therefore $U_n^{(j+1)}$ has real and simple roots, by Lemma 1.10 in \cite{F}. %Furthermore, since $U_{j}$ and $W_{j}$ have positive coefficients, so does $U_{j+1}$, hence it has negative roots.  

\medskip
\item[(b)] Consider the recurrence 
$$V_n^{(j+1)}=U_n^{(j+1)}+(n-3j-2)(n-3j-1)(3j+2)(3j+3) V_n^{(j)}.$$
 One can rewrite $U_n^{(j+1)}$ as given in part $(a)$ above. By (2.2), $U_n^{(j)}\leftarrow V_n^{(j)}$ and by (2.3), $\xi  W_n^{(j)}\leftarrow V_n^{(j)}$ since $W_n^{(j)}$ and $V_n^{(j)}$ only have negative and simple roots. 
Proof is concluded by the same argument  as in part $(a)$.

\medskip
\item[(c)] Consider the recurrence  
$$W_n^{(j+1)}=V_n^{(j+1)}+(n-3j-3)(n-3j-2)(3j+3)(3j+4)  W_n^{(j)}.$$
 One can rewrite $V_n^{(j+1)}$ as in part $(b)$. By (2.3), $V_n^{(j)}\leftarrow W_n^{(j)}$. $U_n^{(j+1)}\leftarrow W_n^{(j)}$ by part $(a)$, Corollary 1.30 in \cite{F} and  the fact that both $U_n^{(j+1)}$ and $W_n^{(j)}$ have only negative and simple roots. 
Proof is concluded by the same argument  as in parts $(a)$ and $(b)$.
\end{enumerate}
\end{proof}
%
%\smallskip
%Finally, let us show that all roots of the considered polynomials are simple. Assume that $U_{j}$ has a double root $p$. Then by Corollary 1.30 in \cite{F}, $p$ should also be a root of $\xi  W_{j-1}$ and $U_{j-1}$. 
%Hence when  we conduct our induction on $j$ and assume that $\xi  W_{j-1}$ and $U_{j-1}$ have only simple roots,  we get a contradiction since $U_{j}$ is a linear combination of those polynomials. Note that, for $j=2$, we  have only simple zeros from the latter formulas.  
%The same argument applies to $V_{j}$ and $W_{j}$. 

\medskip 
Let us now settle Theorem~\ref{th:main} in the special case $a_n=0;\; n=1,2,3,\dots$. 

\begin{proof}
In order to apply the approach of \cite{KvA}, we need  the variable recurrence coefficients to stabilize when $\frac{k}{n}\to \tau$, for any fixed $\tau\in [0,1]$.  To get such stabilization,  consider the rescaled matrix 
\begin{center}
\hspace{-2cm}
\begin{equation}
\frac{1}{n^{4/3}}M_n^{(0)} -\be I:= \left(
\begin{matrix}
-\be & 0 & \frac{2}{n^{4/3}} & 0 & 0 & \cdots & 0 \\
\frac{n}{n^{4/3}} & -\be & 0 & \frac{6}{n^{4/3}} & 0 & \cdots & 0 \\
0 & \frac{n-1}{n^{4/3}} & -\be & 0 & \frac{12}{n^{4/3}} & \cdots & 0 \\
\vdots & \vdots & \ddots & \ddots & \ddots & \ddots & \vdots \\
0 & 0 & \cdots & \frac{3}{n^{4/3}} & -\be & 0 & \frac{n(n-1)}{n^{4/3}} \\
0 & 0 & \cdots & 0 & \frac{2}{n^{4/3}} & -\be & 0 \\
0 & 0 & \cdots & 0 & 0 & \frac{1}{n^{4/3}} & -\be
\end{matrix}
\right)
\end{equation}\label{NNNNmatrix}
\end{center}
which is obtained from the matrix $M_n^{(0)} -\la I$ defined in \eqref{Nmatrix} dividing it by $n^{4/3}$ and setting $\be=n^{-4/3}\la$. Denote by $\widetilde \De_n^{(k)}(\be)$ the $k$-th principal minor of the above matrix.  Thus, we get $\widetilde \De_n^{(k)}(\be):=\De_n^{(k)}(n^{4/3}\be)/n^{4k/3}$.  The sequence $\{\widetilde \De_n^{(k)}(\be)\}_{k=1}^{n+1}$ satisfies the scaled recurrence 
\begin{equation}\label{recrelmod}
\widetilde \De_n^{(k)}(\be)=-\be\widetilde\De_n^{(k-1)}(\be)+\frac{(n-k+2)(n-k+3)(k-1)(k-2)}{n^4}\widetilde\De_n^{(k-3)}(\be),
\end{equation}
 obtained from \eqref{recrel} by substituting $\be=n^{-4/3}\la$.  In other words, \eqref{recrelmod} is satisfied by the characteristic polynomials of the principal minors of  the matrix $\frac{1}{n^{4/3}}M_n^{(0)}$.  
Then when $\frac{k}{n}\to\tau,$  the latter recurrence transforms  into the recurrence:
\begin{equation}\label{recrelstab}
\Omega_\tau^{(k)}(\be)=-\be\Omega_\tau^{(k-1)}(\be)+(1-\tau)^2\tau^2\Omega_\tau^{(k-3)}(\be) 
\end{equation}
 with constant coefficients. The polynomial $\Omega_\tau^{(k)}(\be)$ can be interpreted at the $k$-th principal minor of the infinite Toeplitz matrix
 \begin{center}
\hspace{-2cm}
\begin{equation}
 \left(
\begin{matrix}
-\be & 0 & \tau^2 & 0 & 0 & \cdots & 0 \\
(1-\tau) & -\be & 0 & \tau^2 & 0 & \cdots & 0 \\
0 & (1-\tau)& -\be & 0 & \tau^2 & \cdots & 0 \\
\vdots & \vdots & \ddots & \ddots & \ddots & \ddots & \vdots \\
0 & 0 & \cdots & (1-\tau) & -\be & 0 & \tau^2\\
0 & 0 & \cdots & 0 & (1-\tau) & -\be & 0 \\
0 & 0 & \cdots & 0 & 0 & (1-\tau) & -\be
\end{matrix}
\right).
\end{equation}\label{TTmatrix}
\end{center}

\medskip
Following \cite{KvA} and using Proposition~\ref{conj:2}, we conclude that the Cauchy transform %(considered in the complement $\bC\setminus \Omega$ where $\Omega$ is an appropriate compact subset) 
  of the asymptotic root-counting measure for the polynomial sequence $$\left\{Sp_n(0,\be n^{4/3})\right\}=\left\{n^{\frac{4(n+1)}{3}}\cdot \widetilde\De_n^{(n+1)}(\be)\right\}$$ %:=\left\{Sp_n(0,\frac{\la}{n^{4/3}})\right\}
 is obtained by averaging the Cauchy transforms of the asymptotic distributions of \eqref{recrelstab} over $\tau\in [0,1]$.% since we already know that these characteristic polynomials are real-rooted. 

\smallskip
In fact, in the case under consideration even the density of the former distribution can be obtained by averaging the densities of the latter family of distributions which we can confirm as follows. 

\smallskip
Recurrence \eqref{recrelstab} is similar to the one considered in the last section of \cite{BBSh} and has very nice asymptotic distribution of its roots, see Figure~\ref{figTau}. 
Observe that for any $\tau\in[0,1],$ the initial conditions for \eqref{recrelstab} are given by $$\Omega_\tau^{(-2)}(\be)=\Omega_\tau^{(-1)}(\be)=0,\; \Omega_\tau^{(0)}(\be)=1.$$ 
 
Since \eqref{recrelstab} has constant coefficients, the support of the asymptotic root-counting measure of its solution is described by  the well-known result of Beraha-Cahane-Weiss \cite{BKW2}. Namely, this support  coincides with the set of all $\be\in \bC$ such that among three solutions of the characteristic equation
\begin{equation}\label{char}
\Psi^3+\be \Psi^2-(1-\tau)^2\tau^2=0
\end{equation}
 with respect to the variable $\Psi$ two have  the same modulus  which is bigger or equal to   the modulus of the third solution of \eqref{char}. From considerations of \cite {BBSh} one can easily derive that, for any fixed $\tau\in [0,1],$ this support is the union of three intervals starting at the origin and ending at the branching points of \eqref{char}, i.e., those values of $\be$ and $\tau$ for which  \eqref{char} has a multiple root with respect to  $\Psi$.  The latter branching points  are given by the equation:
\begin{equation}\label{branch}
\be^3=\frac{27}{4}(1-\tau)^2\tau^2
\end{equation}
and their location for three values of $\tau$ is shown in Figure~\ref{figTau}.  Observe that if for $\tau\in [0,1]$,  we denote the branching point lying on the positive half-axis by $\be^+(\tau)$, then it attains its maximal value when $\tau=1/2$ and this maximum equals $\frac{3}{4}$.

As we mentioned before, in the case under consideration for any $\tau\in [0,1]$, the roots of all polynomials generated by \eqref{recrelstab} lie on three fixed rays through the origin. Therefore, the density of the asymptotic root-counting measure  of the polynomial sequence
 $$\left\{Sp_n(0,\be n^{4/3})\right\}=\left\{n^{{\frac{4(n+1)}{3}}}\widetilde\De_n^{(n+1)}(\be)\right\}$$  %:=\left\{Sp_n(0,\frac{\la}{n^{4/3}})\right\} 
 is obtained by averaging the densities  of \eqref{recrelstab} over $\tau\in [0,1]$, cf. \cite{CCvA} and \cite {KvA}. Therefore,  the support of the asymptotic root distribution for  $\{Sp_n(0,\be n^{4/3})\}$ in the $\be$-plane is the union of three intervals of length $\frac{3}{4}$. 
\end{proof}

\begin{figure}

\begin{center}
\includegraphics[width=0.25\textwidth]{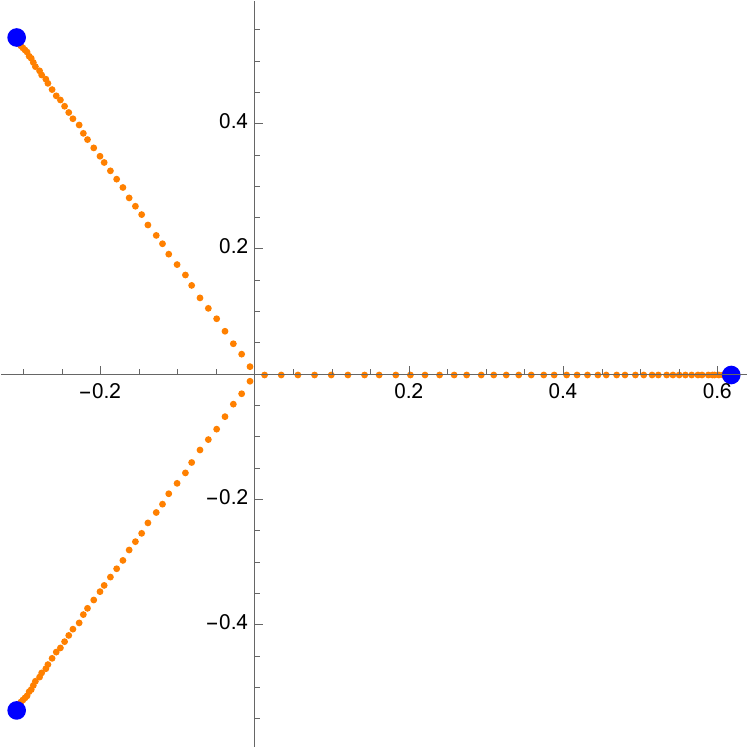} \quad \quad \includegraphics [scale=0.25]{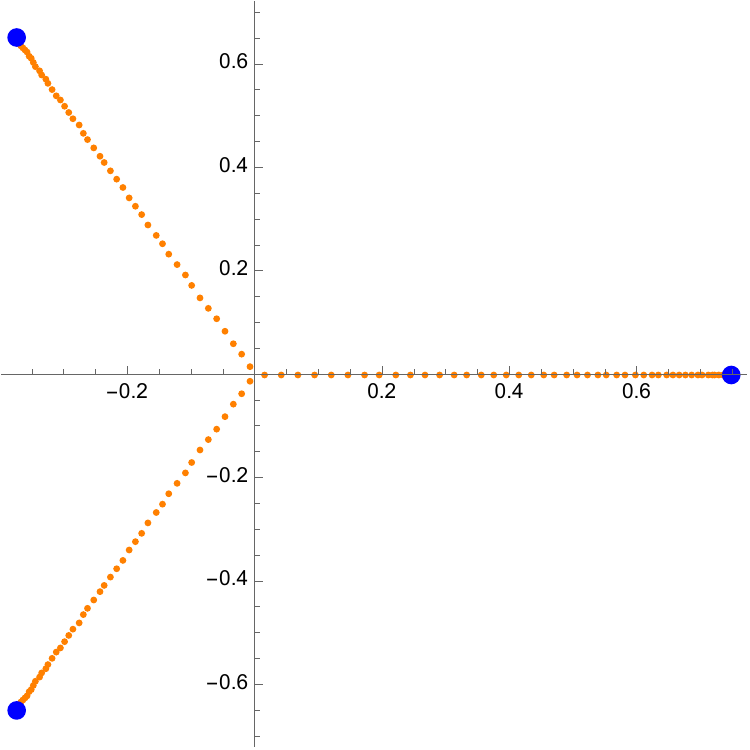}  \quad \quad \includegraphics [scale=0.25]{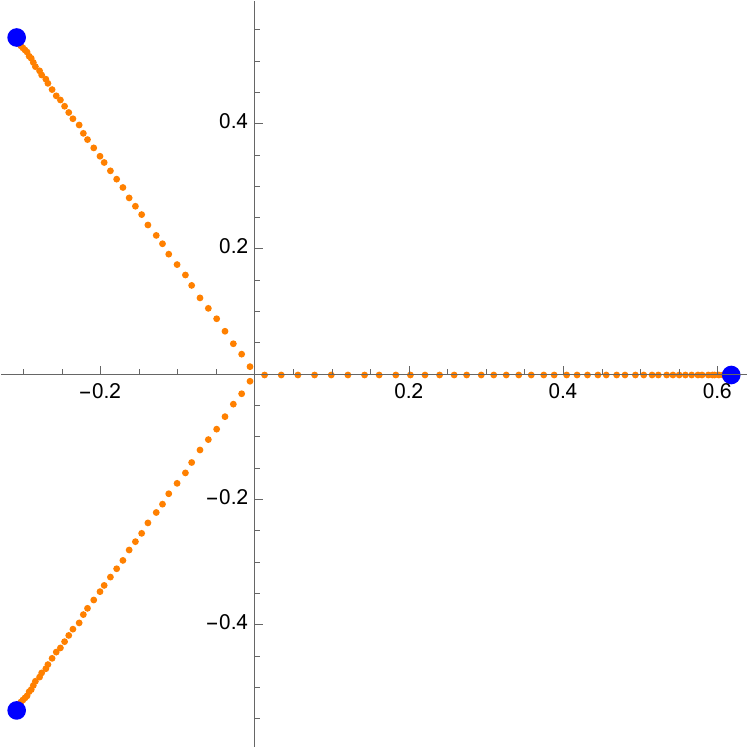}
\end{center}

%\vskip 1cm

\caption{Roots of $\Omega_\tau^{(150)}(\beta)$ for $\tau=1/4$, $\tau=1/2$,  $\tau=3/4$. Endpoints  of the segments are the branching points given by \eqref{branch} for the respective values of parameter $\tau$.}
\label{figTau}
\end{figure}

\begin{proof}[Proof of Theorem~\ref{th:main} in case  when $\lim_{n\to\infty}\frac{a_n}{n^{2/3}}=0$,] We need to show that Proposition~\ref{conj:2} holds asymptotically for any sequence $\{a_n\}$ of complex numbers satisfying the condition $$\lim_{n\to\infty}\frac{a_n}{n^{2/3}}=0.$$  

\smallskip
Indeed, the principal minors of the matrix  
\begin{center}
\begin{equation}\label{NNNmatrix}
M_n^{(a_n)} -\la I:= \left(
\begin{matrix}
-\la & a_n & 2 & 0 & 0 & \cdots & 0 \\
n & -\la & 2a_n & 6 & 0 & \cdots & 0 \\
0 & n-1 & -\la & 3a_n & 12 & \cdots & 0 \\
\vdots & \vdots & \ddots & \ddots & \ddots & \ddots & \vdots \\
0 & 0 & \cdots & 3 & -\la & (n-1)a_n & n(n-1) \\
0 & 0 & \cdots & 0 & 2 & -\la & na_n \\
0 & 0 & \cdots & 0 & 0 & 1 & -\la
\end{matrix}
\right),
\end{equation}
\end{center}
satisfy the recurrence 
$$\De_n^{(k)}(a_n,\la)=-\la\De_n^{(k-1)}(a_n,\la)+(k-1)(n-k+2)a_n \De_n^{(k-2)}(a_n,\la)$$ 
\begin{equation}\label{recrelN}
+(n-k+2)(n-k+3)(k-1)(k-2)\De_n^{(k-3)}(a_n,\la),
\end{equation}
where $k=1, 2, \dots, n+1,$ with the initial conditions 
$$\De_n^{(-2)}(a_n,\la)=\De_n^{(-1)}(a_n,\la)=0, \quad \De_n^{(0)}(a_n,\la)=1.$$

To obtain a converging sequence of   root-counting measures one has to consider  the scaled matrix 
$\frac{1}{n^{4/3}}M_n^{(a_n)} -\be I$ given by:
\begin{center}
\hspace{-2cm}
\begin{equation}\label{MMMmatrix}
 \left(
\begin{matrix}
-\be &  \frac{a_n}{n^{4/3}} & \frac{2}{n^{4/3}} & 0 & 0 & \cdots & 0 \\
\frac{n}{n^{4/3}} & -\be &  \frac{2a_n}{n^{4/3}} & \frac{6}{n^{4/3}} & 0 & \cdots & 0 \\
0 & \frac{n-1}{n^{4/3}} & -\be &  \frac{3a_n}{n^{4/3}} & \frac{12}{n^{4/3}} & \cdots & 0 \\
\vdots & \vdots & \ddots & \ddots & \ddots & \ddots & \vdots \\
0 & 0 & \cdots & \frac{3}{n^{4/3}} & -\be &  \frac{(n-1)a_n}{n^{4/3}} & \frac{n(n-1)}{n^{4/3}} \\
0 & 0 & \cdots & 0 & \frac{2}{n^{4/3}} & -\be &  \frac{na_n}{n^{4/3}} \\
0 & 0 & \cdots & 0 & 0 & \frac{1}{n^{4/3}} & -\be
\end{matrix}
\right).
\end{equation}
\end{center}
It is obtained by dividing  \eqref{NNNmatrix} by $n^{4/3}$ and setting $\be=\la/ n^{4/3}$. Its principal minors satisfy the recurrence 
$$
\widetilde \De_n^{(k)}(a_n,\be)=-\be\widetilde\De_n^{(k-1)}(a_n,\be)-\frac{(k-1)(n-k+2)a_n}{n^{8/3}}\widetilde\De_n^{(k-2)}(a_n,\be)
$$
\begin{equation}\label{recrelmodN}
+\frac{(n-k+2)(n-k+3)(k-1)(k-2)}{n^4}\widetilde\De_n^{(k-3)}(a_n,\be). 
\end{equation}
 
\medskip
Since $\lim_{n\to\infty}\frac{a_n}{n^{2/3}}=0$, then for $\frac{k}{n}\to \tau,$ the family \eqref{recrelmodN}  of recurrence relations converges to the earlier family \eqref{recrelstab} 
corresponding to  the case $a_n=0, \; n=1,2,3,\dots$. Therefore the measure obtained by averaging the family of root-counting measures for  recurrence relations with constant coefficients is exactly the same as in the previous case $a_n=0$, i.e., it coincides with $\nu_0$. Additionally observe that by a general result of \cite{KvA}, the support of the asymptotic root-counting measure of the sequence 
 $\{\widetilde \De_n^{(k)}(a_n,\be)\}$ can only be smaller than that of $\nu_0$ and their Cauchy transforms must coincide outside the support of $\nu_0$.  Since the support of $\nu_0$ is  the union of three straight intervals  through the origin  the resulting asymptotic root-counting measure for the sequence $\{Sp_n(a_n,\be n^{4/3})\}$ in case   $\lim_{n\to\infty}\frac{a_n}{n^{2/3}}=0$ coincides with $\nu_0$ as well.
\end{proof} 

\begin{remark}
{\rm Since the support of the limiting measure $\nu_0$ consists of three segments through the origin it is in principle possible  to find integral formulas for the density and the Cauchy transform of $\nu_0$ similar to those presented in \cite{KvA}, \cite{CCvA} and \cite {ShTa}. In particular,  in the complement to the support of $\nu_0$, its Cauchy transform is given by
$$\C_{\nu_0}(\be)=\int_0^1\frac{\partial}{\partial \beta} \left(\log \tilde \Psi \right) {d\tau} ,$$
where $\tilde\Psi$ is the unique solution of \eqref{char} satisfying  $\lim_{\be\to \infty}\frac{\tilde\Psi}{\be}=-1$.  
However it seems difficult to find either a somewhat explicit expression for $\C_{\nu_0}(\be)$ or a linear differential operator with polynomial coefficients annihilating $\C_{\nu_0}(\be)$.   % which  in a similar case th  authors succeeded to obtain in \cite{ShTa}. 
   (Observe that such an operator always exists since $\C_{\nu_0}(\be)$ belongs to  the Nilsson class, see \cite{Nil}). 
}
\end{remark}

\section {Case  $\lim_{n\to\infty}\frac{a_n}{n^{2/3}}=A\neq 0$.}\label{sec:phys}

\subsection {``Proof" of Proposition~\ref{quart}}\hfill\\ 

Similarly to the previous section, let $\De^{(k)}_n(a_n,\la)$ be the characteristic polynomial of the $k$-th principal minors of $M_n^{(a_n)}$, see \eqref{matrix}.  As in the previous section let us start with a special sequence $\{a_n=A n^{2/3}\}$ for some fixed $A\neq 0$. Next consider the matrix  

\begin{center}
\hspace{-2cm}
\begin{equation}\label{NAmatrix}
M_n^{(An^{2/3})} -\la I:= \left(
\begin{matrix}
-\la & An^{2/3} & 2 & 0 & 0 & \cdots & 0 \\
n & -\la & 2An^{2/3} & 6 & 0 & \cdots & 0 \\
0 & n-1 & -\la & 3An^{2/3} & 12 & \cdots & 0 \\
\vdots & \vdots & \ddots & \ddots & \ddots & \ddots & \vdots \\
0 & 0 & \cdots & 3 & -\la & (n-1)An^{2/3} & n(n-1) \\
0 & 0 & \cdots & 0 & 2 & -\la & nA n^{2/3} \\
0 & 0 & \cdots & 0 & 0 & 1 & -\la
\end{matrix}
\right).
\end{equation}
\end{center}
and denote its $k$-th principal minor by $\De_n^{(k)}(An^{2/3},\la)$. This sequence of minors satisfies the recurrence relation of length $4$ of the form:
$$
%\label{recrelgenA}
\De_n^{(k)}(An^{2/3},\la)=-\la\De_n^{(k-1)}(An^{2/3},\la)-(k-1)(n-k+2)An^{2/3}\De^{(k-2)}_n(An^{2/3},\la)$$
$$+(n-k+2)(n-k+3)(k-1)(k-2)\De_n^{(k-3)}(An^{2/3},\la).$$
%\end{equation}

To get  stabilization similar to that of \S~\ref{sec:2},  introduce the scaled matrix 
\begin{center}
\hspace{-2cm}
\begin{equation}\label{NNmatrix}
\frac{1}{n^{4/3}}M_n^{(An^{2/3})} -\be I:= \left(
\begin{matrix}
-\be & \frac{A}{n^{2/3}} & \frac{2}{n^{4/3}} & 0 & 0 & \cdots & 0 \\
\frac{n}{n^{4/3}} & -\be & \frac{2A}{n^{2/3}} & \frac{6}{n^{4/3}} & 0 & \cdots & 0 \\
0 & \frac{n-1}{n^{4/3}} & -\be & \frac{3A}{n^{2/3}} & \frac{12}{n^{4/3}} & \cdots & 0 \\
\vdots & \vdots & \ddots & \ddots & \ddots & \ddots & \vdots \\
0 & 0 & \cdots & \frac{3}{n^{4/3}} & -\be & \frac{(n-1)A}{n^{2/3}} & \frac{n(n-1)}{n^{4/3}} \\
0 & 0 & \cdots & 0 & \frac{2}{n^{4/3}} & -\be & \frac{nA}{n^{2/3}} \\
0 & 0 & \cdots & 0 & 0 & \frac{1}{n^{4/3}} & -\be
\end{matrix}
\right)
\end{equation}
\end{center}
where $\be= \la n^{-4/3}$.

Its principal minors (which we  denote by $\widetilde\De_n^{(k)}(An^{2/3},\be)$) satisfy the relation
\begin{equation}\label{recrelmodA}
\widetilde\De_n^{(k)}(An^{2/3},\be)=-\be\widetilde\De_n^{(k-1)}(An^{2/3},\be)-\frac{(k-1)(n-k+2)A}{n^{2}}\widetilde\De_n^{(k-2)}(An^{2/3},\be)$$
 $$+\frac{(n-k+2)(n-k+3)(k-1)(k-2)}{n^4}\widetilde\De_n^{(k-3)}(An^{2/3},\be).
\end{equation}

Assuming that $\frac{k}{n}\to \tau,$ we obtain that \eqref{recrelmodA} tends to the following relation with constant coefficients:
\begin{equation}\label{recrelstabA}
\Omega_\tau^{(k)}(A,\be)=-\be\Omega_\tau^{(k-1)}(A,\be)-A\tau(1-\tau)\Omega_\tau^{(k-2)}(A,\be)+(1-\tau)^2\tau^2\Omega_\tau^{(k-3)}(A,\be),
\end{equation}
whose characteristic equation is given by
\begin{equation}\label{characA}
\Psi^3+\be \Psi^2+A\tau(1-\tau)\Psi-(1-\tau)^2\tau^2=0.
\end{equation}

The polynomial $\Omega_\tau^{(k)}(A,\be)$ can be interpreted at the $k$-th principal minor of infinite Toeplitz matrix
 \begin{center}
\hspace{-2cm}
\begin{equation}
 \left(
\begin{matrix}
-\be & A\tau & \tau^2 & 0 & 0 & \cdots & 0 \\
(1-\tau) & -\be & A\tau & \tau^2 & 0 & \cdots & 0 \\
0 & (1-\tau)& -\be & A \tau& \tau^2 & \cdots & 0 \\
\vdots & \vdots & \ddots & \ddots & \ddots & \ddots & \vdots \\
0 & 0 & \cdots & (1-\tau) & -\be & A\tau & \tau^2\\
0 & 0 & \cdots & 0 & (1-\tau) & -\be &A\tau \\
0 & 0 & \cdots & 0 & 0 & (1-\tau) & -\be
\end{matrix}
\right).
\end{equation}\label{TTmatrix}
\end{center}

For a generic complex $A$,  the union over all $\tau\in [0,1]$  of the supports of the asymptotic root-counting measures $\mu_A^{(\tau)}$ for the polynomial sequences $\{\Omega_\tau^{(k)}(A,\be)\}_{k=1}^\infty$ (depending on $\tau\in [0,1]$)  is  strictly larger than that of $\nu_A$, see Figure~\ref{figA1}. In this case we can only conclude that the corresponding Cauchy transforms of both measures coincide with each other in the complement to the larger support.  (This fact follows from the complex version of the main result of \cite{KvA} corresponded to the first author by A.~Kuijlaars, see below).   The latter circumstance implies that the measure 
$$\mathfrak M_A=\int_{0}^1 \mu_A^{(\tau)}d\tau$$
is obtained as the balayage of measure $\nu_A$. (The above mentioned complex version  of the main result of \cite{KvA} is worked out in details in \S~\ref{sec:appen}.)
\begin{figure}

\begin{center}
\includegraphics [scale=0.4]{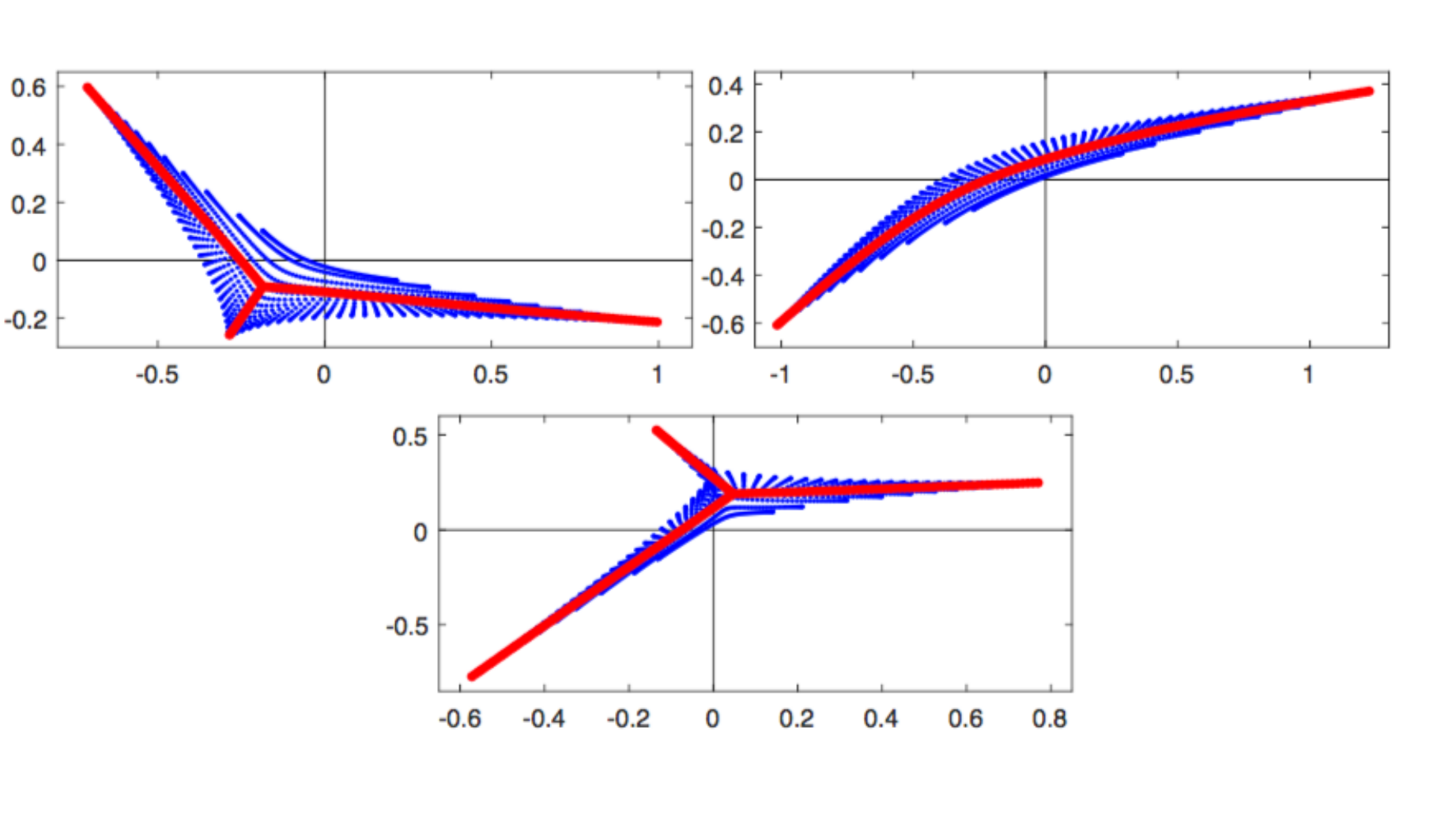}
\end{center}

%\vskip 1cm

\caption{The blue domain is (an approximation to) the union over $\tau\in [0,1]$ of the supports for  the asymptotic root-counting measures for the polynomial sequences $\{\Omega_\tau^{(k)}(A,\be)\}$ defined by  \eqref{recrelstabA}.  The red curve is the zero locus of  $Sp_{200}(An^{2/3},\be n^{4/3})$. We use   $A=(1-i)/2$ (upper left), $A=1+i$ (upper right), and $A=i/2$ (bottom). }
\label{figA1}
\end{figure}

\begin{figure}

\begin{center}
 \includegraphics [scale=0.4]{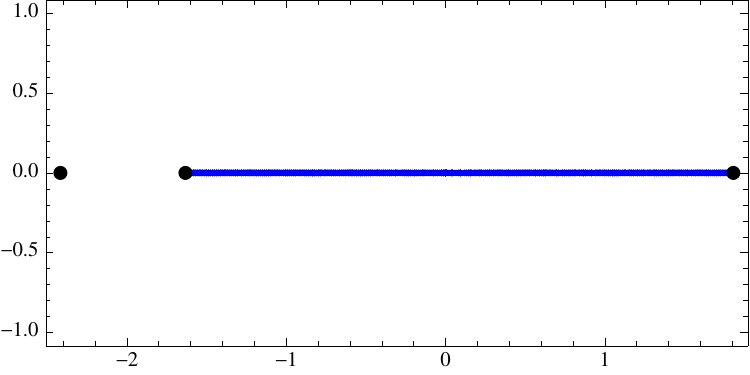} 

\end{center}

%\vskip 1cm

\caption{The union over $\tau\in [0,1]$ of the supports of the asymptotic root-counting measures for the sequences   $\{\Omega_\tau^{(k)}(A,\be)\}$  in case $A=3$.  Black dots are the three branching points given by  \eqref{endss}. }
\label{figA=3}
\end{figure}

\begin{figure}

\begin{center}
  \includegraphics [scale=0.3]{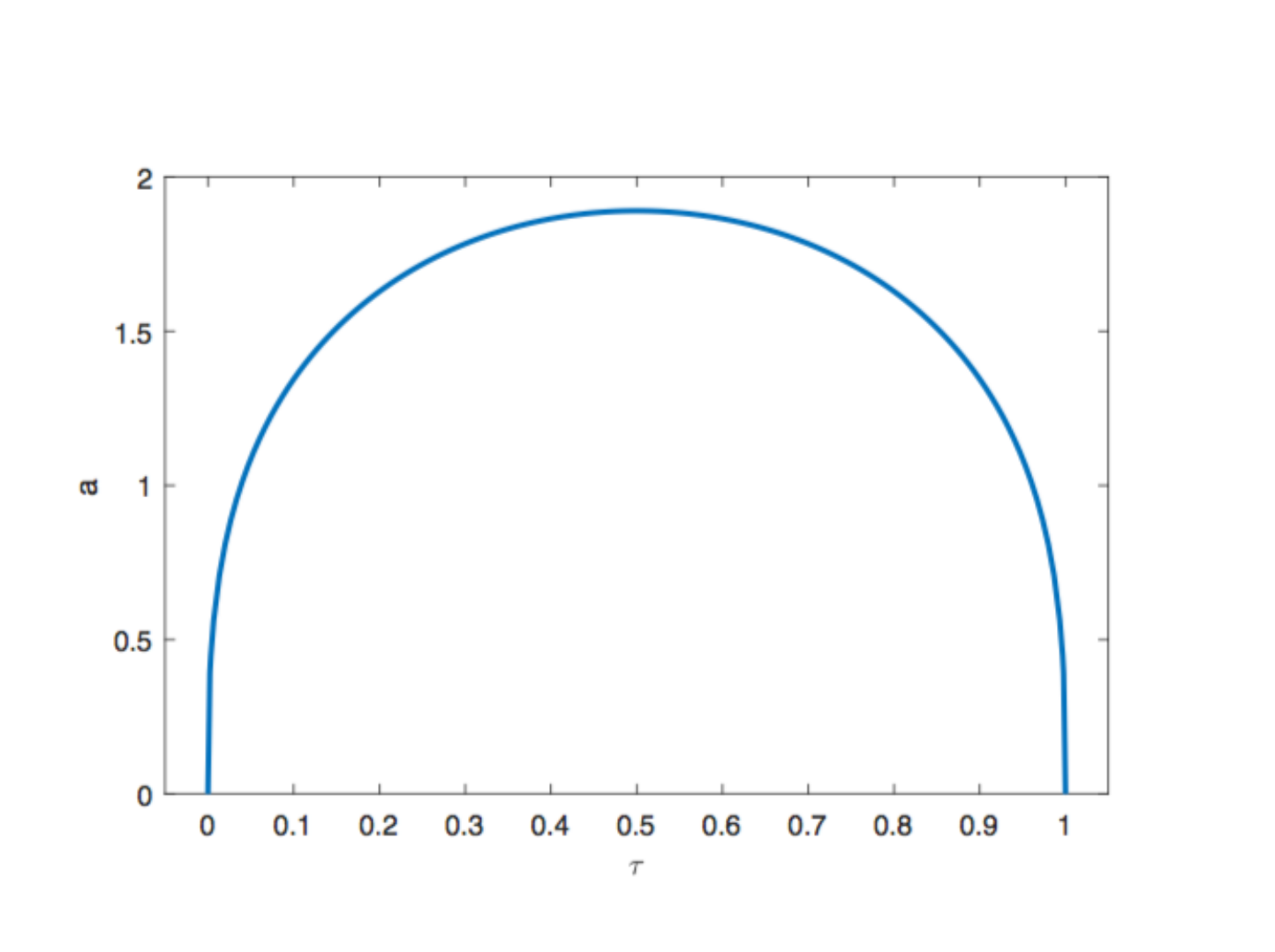} 

\end{center}

%\vskip 1cm

\caption{Plot of curve $A^3-27\tau+27\tau^2=0$ in the real $(A,\tau)$-plane.  $A$ is the vertical and $\tau$ is the horizontal axes.}
\label{figAtau}
\end{figure}

\medskip
However, the next Lemma shows that in case $A\ge \frac{3}{\root {3} \of  {4}}$ (which fits   the situation covered by the main result of \cite{CCvA},) these supports coincide and one can obtain the density of $\nu_A$ by averaging the densities of \eqref{recrelstabA}.

\begin{lemma}\label{apositive}
If $A\ge  \frac{3}{\root {3} \of  {4}},$ then   for any $\tau\in [0,1],$ the support of the asymptotic root-counting measure given by the polynomial sequence defined by \eqref{recrelstabA} is a real interval. This interval connects two  branching points defined by \eqref{eq:brA}.
\end{lemma}

\begin{proof} Using the standard expression for the discriminant,   one can check that all three branching points of \eqref{characA} with respect to $\Psi$ satisfy the  equation: 
\begin{equation}\label{eq:brA}
 4\be^3  + A^2 \be^2- 18 A \be \tau(1 -  \tau)  +\tau(1- \tau)(27\tau^2- 27 \tau-4 A^3 )=0. 
\end{equation}
To check that for $A \ge  \frac{3}{\root {3} \of  {4}},$ and any $\tau\in [0,1],$ all three solutions of \eqref{eq:brA} are real, we calculate the discriminant of \eqref{eq:brA} with respect to $\be$. Again using symbolic manipulations, we get that this discriminant is given by:
$$Dsc:=16 \tau(1-\tau)(A^3-27\tau+27\tau^2)^3.$$
For $\tau\in[0,1],$ the graph of  $Dsc$ in the real $(A,\tau)$-plane is presented in Figure~\ref{figAtau}. One can easily check that the maximal value of $A$ on this graph  is obtained when $\tau=1/2$ and is equal to $\frac{3}{\root {3} \of  {4}}$. This fact implies that $\frac{3}{\root {3} \of  {4}}$ is  the largest real value of $A$ for which roots of \eqref{eq:brA} w.r.t. $\beta$ can become multiple  for some choice of $\tau\in[0,1]$. Moreover checking the location of these roots  for some  value of $A> \frac{3}{\root {3} \of  {4}}$  (for example, for $A=3$ shown in Figure~\ref{figA=3}) we can conclude that all three roots of \eqref{eq:brA} are real for any $\tau \in [0,1]$. The latter circumstance  together with the reality of the situation imply that supports of the asymptotic root counting measure is real for any $\tau\in [0,1]$.  
\end{proof}

\begin{corollary} \label{cor:si} The maximal absolute value of points in $\Si_n$ grows like $\frac{3}{\root {3} \of  {4}} n^{2/3}.$ 
\end{corollary} 

\begin{remark} {\rm The special value $A=  \frac{3}{\root {3} \of  {4}}$ corresponds to the real corner of the asymptotic limiting domain $\mathcal F$ appearing in Conjecture~\ref{conj:Yab}. %However the latter domain  is defined in such a way that its real corner is located at the point $1$ in the real axis, comp. Appendix~\ref{sec:YV}.
} 
\end{remark} 

\begin{remark} {\rm One can additionally show that for $A\ge  \frac{3}{\root {3} \of  {4}},$ the union of the supports (which is the union of all real intervals described in Lemma~\ref{apositive} coincides with the interval bounded by the  two rightmost branching points given by \eqref{endss}, see  Fig~\ref{figA=3}.} 
\end{remark}

\begin{figure}

\begin{center}
\includegraphics [scale=0.4]{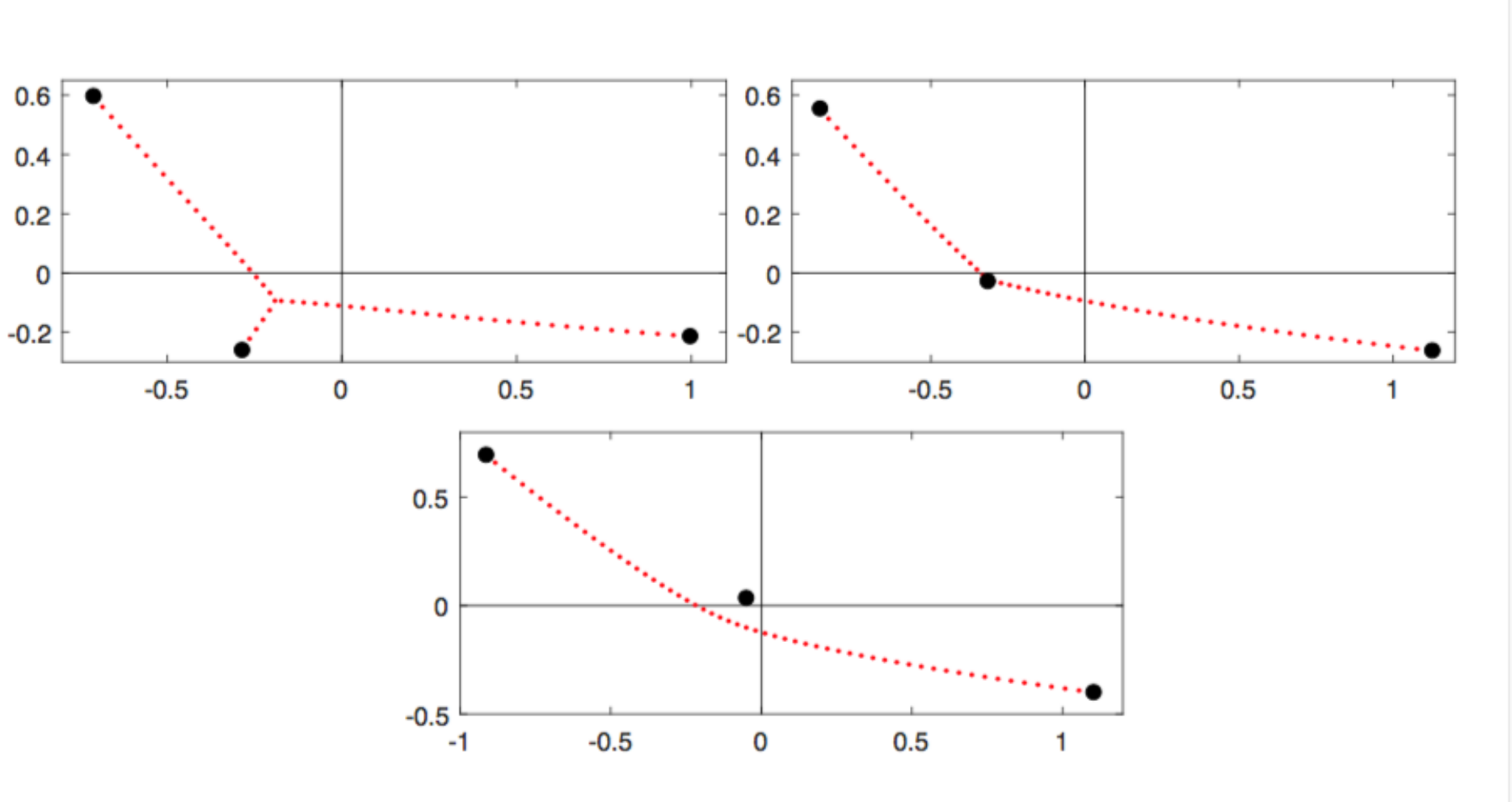}

\end{center}

%\vskip 1cm

\caption{Root distributions of $Sp_{200}(An^{2/3},\be n^{4/3})$  in the $\be$-plane for   $A=(1-i)/2$  (left), $A=4/5-2i/3$ (right)  and  $A=2/3-i\;$ (down). Larger dots are the endpoints of the support given by \eqref{endss}. }
\label{figA}
\end{figure}

\begin{remark}
{\rm Similarly to the case $A=0,$ for any $A\ge  \frac{3}{\root {3} \of  {4}}$,  one can represent the Cauchy transform of $\nu_A$ as 
$$\C_{\nu_A}(\be)=\int_0^1\frac{\partial}{\partial \be} \left( \tilde\Psi \right) d \tau ,$$
 in the complement to the support of $\nu_A$ (which is an interval explicitly given in Lemma~\ref{apositive}). 
%$$\C_{\nu_A}(\be)=\int_0^1\frac{\tilde\Psi^\prime_\be d\tau}{\tilde\Psi},$$
 Here $\tilde\Psi$ is the unique solution of \eqref{characA} satisfying the condition $$\lim_{\be\to \infty}\frac{\tilde\Psi}{\be}=-1.$$
(In an appropriate domain in $\bC$  such presentation for the Cauchy transform of $\nu_A$ is valid for any complex $A$.)  
}
\end{remark}  

Let us now discuss Proposition~\ref{quart}. 

\begin{proof}["Proof`` of Proposition~\ref{quart} under additional convergence assumptions] \hfill\\
 
To obtain the support of the limiting measure $\nu_A$ whose existence is claimed in the Proposition,  we argue as follows. Assume that we have a (sub)sequence $\be_{j_n,n}$ of the eigenvalues of  the sequence of matrices $\{n^{-4/3} M_n^{(An^{2/3})}\}$ (one eigenvalue for each $n$) converging to some finite limit which we denote by $\La$. Denote by $\{p_n\}$ the corresponding (sub)sequence of eigenpolynomials of (the sequence of) differential operators $\{T_n\}$, see \eqref{eq:oper} in the Introduction.  For each $T_n$ the value of its parameter $a_n$ equals $An^{2/3}$. %(These eigenpolynomials are classically referred to as Stieltjes polynomials and  their properties  are  quite well-studied.)   
  Then each eigenpolynomial $p_n$ satisfies its own  differential equation of the form 
 \begin{equation}\label{eq:BLA}
 p_n^{\prime\prime}-(x^2-An^{2/3})p^\prime_n+(nx-\be_{j_n,n} n^{4/3})p_n=0.
 \end{equation}
 
\noindent
{\bf First  assumption.} We assume that if the subsequence  $\{\be_{j_n,n}\}$ has a finite limit  $\La$,  then,  after appropriate scaling of $x$, 
the sequence $\{\mu_n\}$ of the root-counting measures of $\{p_n\}$   converges in the weak sense  to some  limiting measure $\kappa_{a,\La}$ whose support consists of finitely many compact curves and points.

\medskip
This assumption implies that the sequence of appropriately scaled Cauchy transforms of  $\{\mu_n\}$  converges to the Cauchy transform of the limiting measure $\kappa_{a,\La}$. The appropriate scaling of $x$ which might provide such a convergence can be easy  guessed from \eqref{eq:BLA}. Namely, substituting $x=\Theta n^{1/3}$ and  dividing the above equation by $n^{4/3}p_n$, we get  the  relation  
\begin{equation}\label{eq:necess}
\frac{\frac{d^2p_n}{d\Theta^2}} {n^2p_n} - (\Theta^2-A) \frac {\frac{dp_n}{d\Theta}}{np_n}+(\Theta-\be_{j_n,n})=0
\end{equation}
with respect to the new independent variable $\Theta$.  Observe that the scaled logarithmic derivative $\frac{1}{n} \frac{d}{d\Theta} (\log p_n)$ is the Cauchy transform of the root-counting measure of the polynomial $p_n(\Theta n^{1/3})$ with respect to the new variable $\Theta$.

\medskip
\noindent
{\bf Second  assumption.}  Assuming that the sequence $\{\mu_n\}$ of the root-counting measures of $\{p_n(\Theta n^{1/3})\}$ converges to  $\kappa_{a,\La},$ we additionally assume that the sequences of the root-counting measures of  its first and second derivatives  converge to the same  measure $\kappa_{a,\La}.$ 

 (Apparently this assumption  can be rigorously proved by using the same arguments  as presented in \cite{BBSh}.) 

\medskip
Under the above two main assumptions and  using \eqref{eq:necess},  we get that the Cauchy transform $\C_{A,\La}$ of  $\kappa_{A,\La}$ satisfies the quadratic equation: 
\begin{equation}\label{eq:Cauchy}
\C_{A,\La}^2-(\Theta^2-A)\C_{A,\La}+(\Theta-\La)=0
\end{equation}
almost everywhere in $\bC$.  Up to the variable change $x\leftrightarrow \Theta$ and $\La\leftrightarrow \be$ the latter equation coincides with \eqref{quadreq}.

\medskip
\noindent
{\bf Third  assumption.} So far we presented a physics-style argument showing that if a   (sub)sequence $\{\be_{j_n,n}\}$ of the eigenvalues of  the sequence of matrices $\{n^{-4/3} M_n^{(An^{2/3})}\}$ converges to some limit $\La,$ then  there exists a probability measure  $\kappa_{A,\La}$ whose Cauchy transform satisfies \eqref{eq:Cauchy} almost everywhere in the $\Theta$-plane. Our final assumption is that the converse to the latter assumption  is true as well, i.e., for each $\La$ with the above properties there exists an appropriate subsequence $\{\be_{j_n,n}\}$ of eigenvalues  of matrices $\{n^{-4/3} M_n^{(An^{2/3})}\}$ converging to $\La$.  %(At the moment we do not have a mathematically rigorous proof  of this statement but such fact in a very similar situation was  settled in \cite {ShTaTa}.)  

\medskip
Simultaneous application of the above three assumptions settles Proposition~\ref{quart} in the case of the special sequence $a_n=An^{2/3}$, $n=1,2,3, \dots$.   The argument for the general case $\lim_{n\to \infty} a_n=An^{2/3}$ repeats verbatim the one used in the proof of Theorem~\ref{th:main}.  \end{proof}

\subsection{Quadratic equations with polynomial coefficients and quadratic differentials}\hfill \\

To finish this section let us present some additional results about quadratic differentials and signed measures. The next result is a special case of Proposition~9 and Theorem~12 of \cite {BoShN}.

\begin{proposition}\label{lm:horiz} There exists a signed measure $\mu_{A,\La}$ whose Cauchy transform satisfies \eqref{eq:Cauchy} almost everywhere if and only if the set of critical horizontal trajectories of the quadratic differential $$-((\Theta^2-A)^2-4(\Theta-\La))d\Theta^2$$ contains all its turning points, i.e., all roots of $P(\Theta,\La)=(\Theta^2-A)^2-4(\Theta-\La)$. (Here by a critical horizontal trajectory of a quadratic differential we mean its horizontal trajectory which starts and ends at the turning points.) 
\end{proposition}

 We know that the support of $\mu_{A,\La}$ should include all the branching points of \eqref{eq:Cauchy} and consists of the critical horizontal trajectories of \eqref{quadquart}.

%\begin{figure}[h!]%\label{S11}
%\begin{center}
%\includegraphics[scale=0.6]{Stieltjes11-eps-converted-to.pdf}
%\epsfig{file=Stieltjes11.eps, width=\textwidth}
%\end{center}

 %\caption{Mutual positions of $\be_j$'s (gray dots) and roots of Stieltjes
%polynomial (red dots). The green dot marks the corresponding
%$\be_k$, $a=1+i$.}

%\label{figstieltjes}
%\end{figure}

%To be studied systematically. Figure~\ref{figstieltjes} shows a typical pattern.

%Also the interlacing property of spectral and Stieltjes polynomials
%for real $a$ should be studied \cite{BMcM},\cite{BMcMV} (it seems
%that the property is not valid, at least for spectral polynomials).

\begin{lemma}\label{crval} The set of the critical values of the polynomial $P(\Theta,\La)=((\Theta^2-A)^2-4(\Theta-\La))$, i.e.,  the set of all $\La$ for which $P(\Theta, \La)$ has a double root with respect to  $\Theta$ is given by the equation: 
\begin{equation}\label{endss}
4\La^3+A^2\La^2-9A\La/2-A^3-27/16=0.
\end{equation}
\end{lemma}

\begin{proof}
Straight-forward calculations. 
\end{proof}

\begin{corollary}\label{endpoints}
The endpoints of the support of $\nu_A$ are contained among the three roots of  equation \eqref{eq:brA} 
  when   $\tau=1/2$, see Figure~\ref{figA}. This equation  coincides with \eqref{endss} where $\La$ is substituted by $\be$. 
\end{corollary}

\section{On branching points and monodromy of the spectrum}

%\subsection{Monodromy of $Sp_n(a,\la)$}

Observe that, for any positive integer $n$ and generic values of parameter $a$, the roots of  $Sp_n(a,\la)$ with respect to $\la$ %(i.e., the quasi-exactly solvable spectrum) 
   are simple. The latter roots are called the \emph{quasy-exactly solvable spectrum} of the quartic oscillator under consideration. 
   
   Moreover, for  any given $n$, and any 
sufficiently large positive $a$, these roots are real and distinct.  The set $\Si_n \subset \bC$ of branching points of $Sp_n(a,\la)$, i.e., the set of all values of $a$ for which two eigenvalues 
coalesce, has cardinality $\binom {n+1}{2}$.  When plotted these branching points form a regular pattern
in the complex plane shown in  Figures~\ref{BBigtriangle} and \ref{triangle10}. 

\medskip
In this subsection we present our (mostly) numerical results and conjectures about the monodromy of the roots of $Sp_n(a,\la)$ when $a$ runs along different closed paths  in the complement to $\Si_n$ in the $a$-plane.  We start with the following statement.

\begin{figure}[H]

\begin{center}
\includegraphics[scale=0.45]{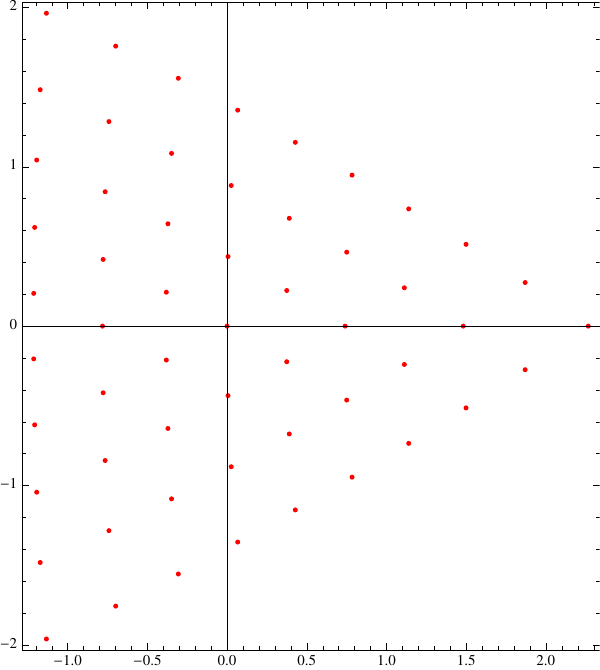}
\end{center}

\caption{The curvilinear triangle of the branching points for $Sp_{10}(a,\la)$.}

\label{triangle10}
\end{figure}

\begin{figure}%\label{aasympt}

\begin{center}
\includegraphics[scale=0.8]{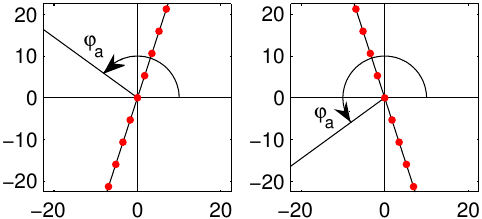}
\end{center}

\caption{Roots of $Sp_8(a,8^{4/3}\be)$. The left figure shows
the situation with $a=500\exp(i\varphi_a), \varphi_a=4\pi/5$, the
right one $a=500\exp(i\varphi_a), \varphi_a=6\pi/5$. In both cases roots are almost uniformly distributed on the interval $[-\sqrt{a}, \sqrt{a}]$.}

\label{figslopes}
\end{figure}

\begin{proposition}\label{prop:asymp} For any given $n$, if $|a|\to \infty$ with $\arg a=\phi$ fixed, then the roots of $Sp_n(a, \la)$  divided by $n^{4/3}$  will be asymptotically uniformly distributed on the straight segment 
$[-\sqrt{a},+\sqrt{a}]$, see Figure~\ref{figslopes}. In particular, if $a$ traverses the circle $R e^{2\pi i t},\; t\in [0,1]$ for any sufficiently large $R,$ then the resulting monodromy of roots of $Sp_n(R,\la )$ (which are all real) is the complete reversing of their order, i.e., the leftmost and the rightmost roots change places, the second from the left and the second from the right change places  etc. 
\end{proposition}

\begin{proof}[Proof of Proposition~\ref{prop:asymp}] Observe that if $a=K e^{2\pi i \phi}$ with $K$ very large then the polynomial $\widetilde {Sp}_n(a,\be):={n^{-\frac{4(n+1)}{3}}} Sp_n(a,\be n^{4/3})$ is coefficient-wise close to $\widehat {Sp}_n(a,\be)$,  where  
$\widehat {Sp}_n(a,\be)$ is the characteristic polynomial of the tridiagonal matrix
\begin{equation}\label{newmatrix}
\widehat M_n^{(a)} := \left(
\begin{matrix}
0 & a/n & 0 & 0 & 0 & \cdots & 0 \\
n/n & 0 & 2a/n & 0 & 0 & \cdots & 0 \\
0 & (n-1)/n & 0 & 3a/n & 0 & \cdots & 0 \\
\vdots & \vdots & \ddots & \ddots & \ddots & \ddots & \vdots \\
0 & 0 & \cdots & 3/n & 0 & (n-1)a/n & 0 \\
0 & 0 & \cdots & 0 & 2/n & 0 & na/n \\
0 & 0 & \cdots & 0 & 0 & 1/n & 0
\end{matrix}
\right).
\end{equation}
(Observe that the above matrix $\widehat M_n^{(a)}$ is tridiagonal and not $4$-diagonal as the previous matrices!)  
To make the situation  more transparent, let us consider the sequence of characteristic polynomials of $\frac{1}{\sqrt{a}}\widetilde M_n^{(a)}$ and of $\frac{1}{\sqrt{a}}\widehat M_n^{(a)},$  where 
\begin{equation}\label{newaddmatrix}
\widetilde M_n^{(a)} := \left(
\begin{matrix}
0 & a/n & 2/n & 0 & 0 & \cdots & 0 \\
n/n & 0 & 2a/n & 5/n & 0 & \cdots & 0 \\
0 & (n-1)/n & 0 & 3a/n & 12/n & \cdots & 0 \\
\vdots & \vdots & \ddots & \ddots & \ddots & \ddots & \vdots \\
0 & 0 & \cdots & 3/n & 0 & (n-1)a/n & n(n-1)/n \\
0 & 0 & \cdots & 0 & 2/n & 0 & na/n \\
0 & 0 & \cdots & 0 & 0 & 1/n & 0
\end{matrix}
\right).
\end{equation}

In other words, we are comparing the roots of $\widetilde {Sp}_n(a,\be)$ divided by $\sqrt{a}$ with that of  $\widehat {Sp}_n(a,\be)$ divided by $\sqrt{a}$. The characteristic polynomials of the respective principal minors of $\frac{1}{\sqrt{a}}\widetilde M_n^{(a)}$ and of $\frac{1}{\sqrt{a}}\widehat M_n^{(a)}$
 satisfy the recurrences:
 \begin{equation}\label{recrelmodAA}
\widetilde\De_n^{(k)}(\ga)=-\ga\widetilde\De_n^{(k-1)}(\ga)-\frac{(k-1)(n-k+2)}{n^{8/3}}\widetilde\De_n^{(k-2)}(\ga)$$
 $$+\frac{(n-k+2)(n-k+3)(k-1)(k-2)}{a^{3/2}n^4}\widetilde\De_n^{(k-3)}(\ga),
\end{equation}
and 
 \begin{equation}\label{recrelmodAAA}
\widehat\De_n^{(k)}(\ga)=-\ga\widehat \De_n^{(k-1)}(\ga)-\frac{(k-1)(n-k+2)}{n^{8/3}}\widehat\De_n^{(k-2)}(\ga), 
\end{equation}
where $\ga=\be/\sqrt{a}$ and both recurrences have the standard boundary conditions: $\widetilde\De_n^{(-1)}(\ga)=\widehat\De_n^{(-1)}(\ga)=0,\; \widetilde\De_n^{(0)}(\ga)=\widehat\De_n^{(0)}(\ga)=1$. As before $\widetilde {Sp}_n(a,\be)=\widetilde\De_n^{(n)}(\ga)$ and $\widehat {Sp}_n(a,\be)=\widehat\De_n^{(n)}(\ga)$. Observe now that, for any fixed $n$ and any $\eps >0$, one can choose $|a|$ so large that each equation in \eqref{recrelmodAA} for $k=1,2,\dots, n$ deviates from the corresponding equation in \eqref{recrelmodAAA} so little that $\widetilde\De_n^{(n)}(\ga)-\widehat\De_n^{(n)}(\ga)$ can be made  coefficientwise arbitrary small.  
 (This can be done due to the presence of $a^{3/2}$ in the denominator of the third term in \eqref{recrelmodAA}).
 
 Now one can easily check by induction that bivariate polynomial $\widehat {Sp}_n(a,\be)$  is quasihomogeneous with weight $1$ for  variable $\be$ and weight $2$ for variable $a$.  When $a$ is positive, then $\widehat {Sp}_n(a,\be)$ is a real-rooted polynomial in $\be$. Since multiplication of $\la$ by $e^{\pi i t}$ and multiplication of $a$ by $e^{2\pi i t}$ multiplies the whole $\widehat {Sp}_n(a,\be)$ by a constant, then for any fixed $a$, the roots of $\widehat {Sp}_n(a,\be)$ with respect to $\be$ lie  on the line through the origin whose slope is half of the slope of $a$. Now consider the roots of $\widehat {Sp}_n(a,\be)$ with respect to $\be$. Observe that if $a=1$, then the spectrum of \eqref{newmatrix} is $(-1, -1+\frac{2}{n}, -1+\frac{4}{n}, \dots, 1-\frac{4}{n}, 1-\frac{2}{n}, 1)$. The same argument as above gives that for any $a\neq 0$, the roots of  $\widehat {Sp}_n(a,\be)$ with respect to $\be$ will be equally spaced on the interval $[-\sqrt{a},\sqrt{a}]$.

Since choosing $|a|$ sufficiently large, we can achieve that all roots of $\widetilde Sp_n^{(n)}(a,\be)$ lie arbitrarily close to those of $\widehat Sp_n^{(n)}(a,\be)$, and since the latter roots are equally distributed on the interval $[e^{-2\pi i \phi},e^{-2\pi i \phi}],$ the result follows. 
 \end{proof}

\begin{figure} %\label{acrit}

\begin{center}
\includegraphics[scale=0.25]{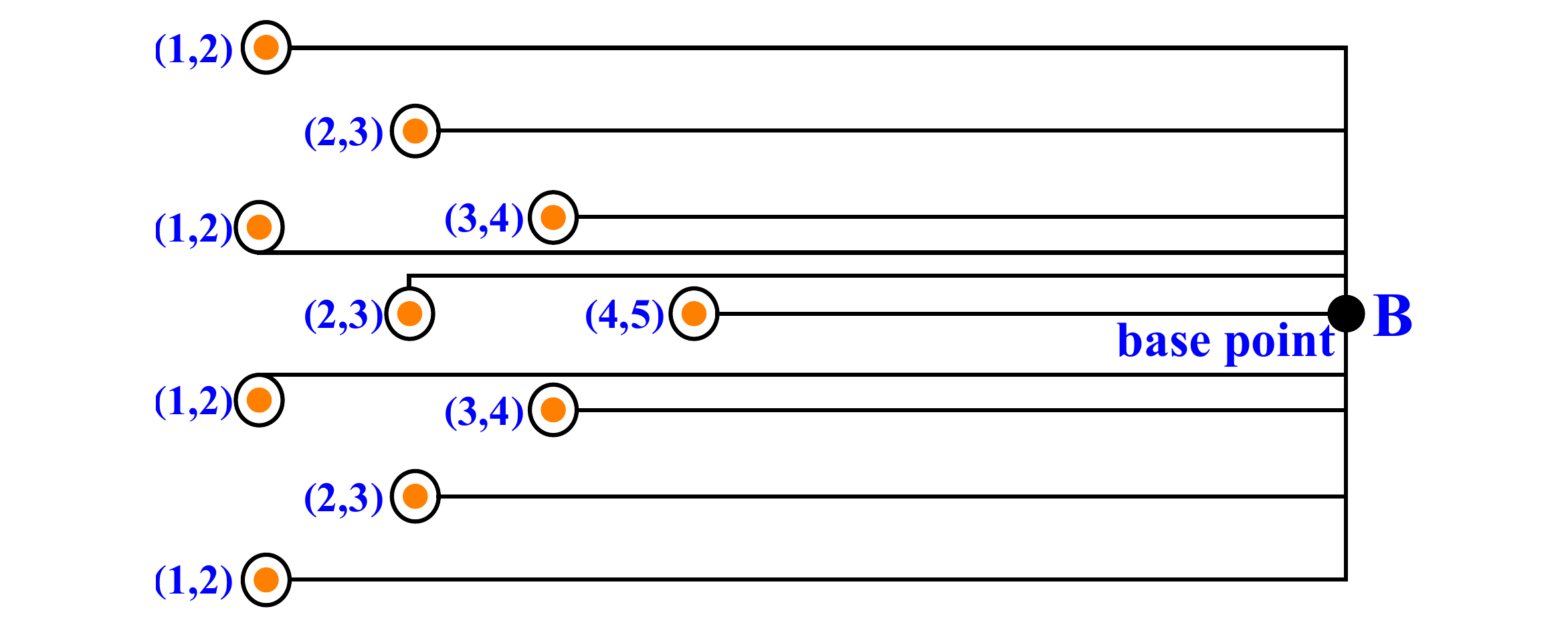}
\end{center}

\caption{The system of standard paths and the monodromy (transpositions) of the spectrum which they produce for $n=4$.}

\label{paths}
\end{figure}

To describe (our conjecture on) the monodromy  of the spectrum, let us introduce a system of standard paths connecting a base point  chosen as  a sufficiently large positive number with every branching point, see Fig.~\ref{paths}.  Based on our numerical experiments, we see that $\Si_n$ form a triangular shape  with points regularly arranged into columns and rows in $\bC$. There are $n$ columns (enumerated from left to right) where the $j$-th column consists of $n-j$ branching points with approximately the same real part and there are $n$ rows (enumerated from bottom to top) where the $i$-th row consists of points   with approximately the same imaginary part.  We denote the branching points $\si_{i,j}\in \Si_n$ where $i=1,\dots, 2n-1$ is the row number and $j=1,\dots, n$ is the column number.  

Fixing a base point $B$ as a sufficiently large positive number,  connect $B$ with every $\si_{i,j}$ by a "vertical hook`` $\mathcal P_{i,j}$, i.e., move from $B$ vertically to the imaginary part of $\si_{i,j}$, then move horizontally to the left untill you almost hit $\si_{i,j}$, then circumgo $\si_{i,j}$ counterclockwise along a small circle centered at $\si_{i,j}$ and return back to $B$ along the same path. Conjecturally, along such a path one will never hit any other branching points unless $\si_{i,j}$ lies on the real axis. In other words,  the imaginary parts of all branching points except for the real ones are all distinct. In case when $\si_{i,j}$ is real one can slightly deform the suggested path (which is a real interval) in an arbitrary way to move it away from the real axis. The resulting monodromy will (conjecturally) be independent of any such small deformation, see below. Finally we can state our surprisingly simple guess. 

\begin{conjecture}\label{conj:monod}
For any $\si_{i,j}\in \Si_n$ and any sufficiently large positive base point $B$,  the monodromy corresponding to the standard path $\mathcal P_{i,j}$ is a simple transposition $(j,j+1)$ of the  roots 
of $Sp_n(B,\la)$ ordered from left to right. (Recall that by our choice of $B$ all roots of $Sp_n(B,\la)$  are real and therefore naturally ordered.)
\end{conjecture}

This conjecture has been numerically checked  for all $n\le 10$. Observe that since the system of  standard paths gives a basis of the fundamental group $\pi_1(\bC\setminus \Si_n)$, then knowing the monodromy for the standard paths, one can  calculate the monodromy along any loop in $\bC\setminus \Si_n$ based at $B$.

\section{Appendix I. Estimates for largest roots of Yablonskii-Vorob'ev polynomials}\label{sec:YV}

In this  section we present both estimates  from above and from below as well as the asymptotic behaviour
for roots of maximal modulus of Yablonskii-Vorob'ev polynomials.

\smallskip
The Yablonskii-Vorob'ev polynomials satisfy the differential-difference relation
\begin{equation}
Q_{n+1}(t)=\frac{tQ_n(t)^2-4(Q_n(t)Q_n^{''}(t)-(Q_n^{'}(t))^2)}{Q_{n-1}(t)}
\end{equation}
with $Q_0(t)=1$ and $Q_1(t)=t$. The roots of $Q_n$ approximately cover  a 
triangular shape invariant under rotation by $2\pi/3$ whose edges are  curves instead of  straight lines. The whole pattern is
invariant under the dihedral group of symmetries of an equilateral triangle. The roots of maximal modulus lie on a circle with centre the origin.
These and other properties are listed in \cite{ClMa}.

Our strategy is to connect roots and coefficients of Yablonskii-Vorob'ev polynomials.
Similar approach was used in \cite{Ka}. Our results are sharper and we also derive the
asymptotic formula. As the first step we form rational functions
$$
\rho_n(t)=(\log(Q_n(t)))^{''}-t/4.
$$
Another relation between $\rho_n$ and $Q_n$ is
$$
\rho_{n+1}(t)=-Q_n(t)Q_{n+2}(t)/(4Q_{n+1}(t)^2).
$$
From this and from
$$
Q_n(t)=\prod_{k=1}^{n(n+1)/2}(t-\alpha_{n,k})=\sum_{k=0}^{n(n+1)/2}c_{n,k}t^k
$$
we see that the Laurent expansion at $\infty$ for $\rho_n$ is convergent in $|t|>A_{n+1}$,
where $A_n:=\max{|\alpha_{n,k}|}_{1\leq k\leq n(n+1)/2}$:
\begin{equation}
\rho_n(t)=-\frac{t}{4}-\sum_{j=0}^{\infty} (-1)^j \rho_{n,j} t^{-(3j+2)}.
\end{equation}
A closer inspection reveals that
\begin{equation}\label{psnj}
\frac{\rho_{n,j}}{3j+1}=(-1)^j s_{n,3j},
\end{equation}
where $s_{n,j}$ is the sum of $j$th powers of roots of $Q_n$, \ie
$$
s_{n,j}=\sum_{k=1}^{n(n+1)/2} \alpha_{n,k}^j.
$$

The second step is to find the relation between $s_{n,j}$ and $c_{n,k}$, \ie
\begin{equation}\label{snj}
s_{n,j}=-j c_{n,n(n+1)/2-j}-\sum_{i=1}^{j-1} c_{n,n(n+1)/2-i} s_{n,j-i}.
\end{equation}
To this end it is suitable to rewrite $Q_n(t)$ as
\begin{equation}
Q_n(t)=\sum_{k=0}^{[n(n+1)/6]+1} q_k(n) t^{(n^2+n-6k)/2}.
\end{equation}
We can find explicit form of $q_n$'s:
\begin{equation}\label{cfc}
\begin{aligned}
q_0(n)=&1 \\
q_1(n)=&(n+2)(n+1)n(n-1)/6 \\
q_2(n)=&(n+5)(n+3)(n+2)(n+1)n(n-1)(n-2)(n-4)/72 \\
q_3(n)=&(n^4+2n^3-57n^2-58n+1120)(n+4)(n+3)(n+2)(n+1) \\
 &n(n-1)(n-2)(n-3)/1296 \\
 &\cdots
\end{aligned}
\end{equation}
We get $s_{n,j}$ from \eqref{snj}. We see that $s_{n,j}\neq 0$ only if $j\equiv 0\, \mathrm{mod}\, 3$.
\begin{equation}\label{pfc}
\begin{aligned}
s_{n,3}=&-(n+2)(n+1)n(n-1)/2 \\
s_{n,6}=&2(n^2+n-5)(n+2)(n+1)n(n-1) \\
s_{n,9}=&-4(n^2+n-7)(3n^2+3n-20)(n+2)(n+1)n(n-1) \\
s_{n,12}=&8(11n^6+33n^5-259n^4-573n^3+2348n^2+2640n-7700) \\
 &(n+2)(n+1)n(n-1) \\
 &\cdots
\end{aligned}
\end{equation}
Now, we can limit $|A_n^{3k}|$.  Namely, %It is clear that $s_{n,3k}$ does not exceed $n(n+1)|A_n^{3k}|/2$, \ie
\begin{equation}\label{up}
|A_n|\leq \sqrt[3k]{s_{n,3k}}, \qquad k=1,2,\ldots .
\end{equation}
%and clearly $\sqrt[3k]{s_{n,3k}}\leq\sqrt[3(k-1)]{s_{n,3(k-1)}}$.

On the other hand, it is clear that $s_{n,3k}$ does not exceed $n(n+1)|A_n^{3k}|/2$ and from \eqref{psnj} follows
\eqref{low}

%On the other hand, from \eqref{psnj}
\begin{equation}\label{low}
|A_n|\geq\sqrt[3k]{\frac{2\rho_{n,k}}{(3k+1)n(n+1)}}.
\end{equation}
\begin{figure}%[t]
\centering
\includegraphics[angle=0,height=7cm]{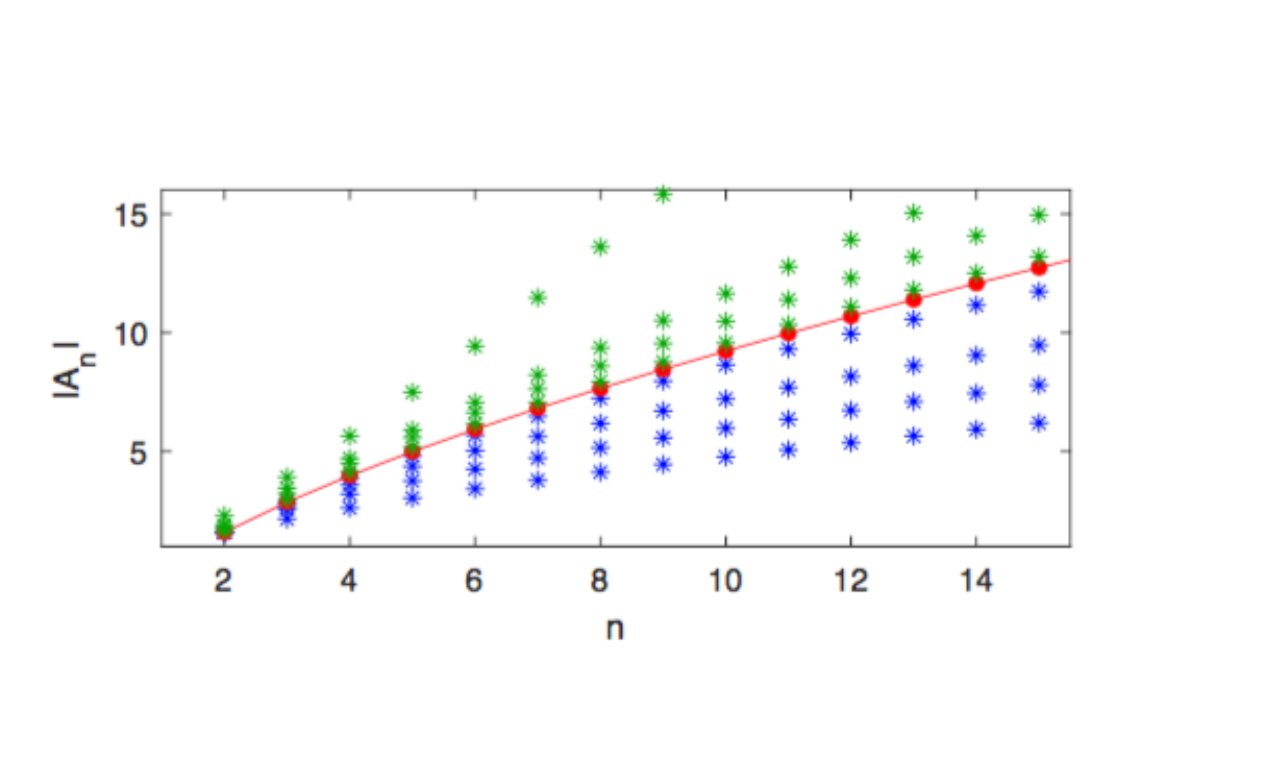}
\caption{Dependence of $|A_n|$ on $n$ (red dots) and the upper \eqref{up} ($k=1,2,3,10$)
(green asterisks) and lower \eqref{low} ($k=1,2,4,15$) (blue asterisks) limits.}
\label{approxfig}
\end{figure}
Again,
$$
\sqrt[3k]{\frac{2\rho_{n,k}}{(3k+1)n(n+1)}}\geq\sqrt[3(k-1)]{\frac{2\rho_{n,k-1}}{(3k-2)n(n+1)}},
$$
so that we have an infinite series of both upper and lower bounds on $|A_n|$. This makes
possible to find the asymptotic behaviour of $|A_n|$ when $n\rightarrow\infty$.

\smallskip
Conserving only the leading term of $s_{n,3k}$ we arrive at a sequence
$$
-\frac{n^4}{2},2 n^6,-12 n^8,88 n^{10},-728 n^{12},6528 n^{14},-62016 n^{16},615296 n^{18},-6314880 n^{20},\ldots
$$
The general form of this sequence is
$$
v_k = \frac{3 (-2)^k \Gamma (3 k)}{\Gamma (k) \Gamma (2 k+3)} n^{2 k+2}
%v_k=\frac{3 (-2)^k \Gamma (3 k)}{\Gamma (k) \Gamma (2 k+3)}
$$
and 

%$$
%n^2, 4n^4, 24n^6, 176n^8, 1456n^{10}, 13056n^{12}, 124032n^{14}, 1230592n^{16}, \ldots
%$$
%I The general term is
%$$
%v_k=\frac{(-1)^k 2^k\Gamma(3k+1)}{\Gamma(k+2)\Gamma(2k+2)}
%$$
%and
$$
\lim_{k\rightarrow\infty} \sqrt[3k]{v_k}=\sqrt[3]{-1}\frac{3}{\sqrt[3]{2}}n^{2/3}.
$$
Besides, the factor $\sqrt[3]{-1}$ shows the direction in which the three roots of
maximal modulus move when $n\rightarrow\infty$.

Numerical results show the rate of convergence to the asymptotic value, cf.  table.

\begin{center}
  \begin{tabular}{ | r | c | c | }
    \hline
    $n$ & $|A_n|$ & $3n^{2/3}/\sqrt[3]{2}-|A_n|$ \\ \hline \hline
    10 & 9.226620959867741 & 1.82547 \\ \hline
    20 & 15.85575092198829 & 1.68836 \\ \hline
    30 & 21.37667838759338 & 1.61260 \\ \hline
    40 & 26.28885026360517 & 1.56068 \\ \hline
    50 & 30.79511630027872 & 1.52140 \\ \hline
    60 & 35.00327495653731 & 1.48994 \\ \hline
    70 & 38.97923077181361 & 1.46376 \\ \hline
    80 & 42.76697737664885 & 1.44140 \\ \hline
    90 & 46.39772099260680 & 1.42191 \\ \hline
    100 & 49.89460772047459 & 1.40467 \\ \hline
    110 & 53.27540358774959 & 1.38922 \\
    \hline
  \end{tabular}
\end{center}

\bigskip
\section{Appendix II. Complex generalization  of the main result of \cite{KvA}} \label{sec:appen} 

The main result of this section is Proposition~\ref{prop:complex} mentioned in \S~\ref{sec:phys} which undoubtedly has independent interest.  For the sake of completeness and for the convenience of our readers we include its proof below.  Here  
we adopt the notation used in paper \cite{KvA} by A.~Kuijlaars and W. Van Assche:
\[
 \lim_{n/N\to t}X_{n,N}=X
\]
which denotes that for the doubly indexed sequence it holds that
\[
 \lim_{j\to\infty}X_{n_j,N_j}=X,
\]
for any $\{n_j\}_{j\in\N},\ \{N_j\}_{j\in\N}\subset\N$ such that $N_j\to\infty$ and $n_j/N_j\to t$, as $j\to\infty$. Additionally, we occasionally add the meaning of the notion of a limit. 
For example, 
\[
 \wlim_{n/N\to t}\mu_{n,N}={\mu}
\]
expresses the limit of the double indexed sequence of measures $\mu_{n,N}$ converging to $\mu$ in the weak$^{\star}$ topology. Similarly,
\[
 \slim_{n/N\to t}A_{n,N}=A
\]
stands for the limit of the double indexed sequence of bounded operators $A_{n,N}$ converging to $A$ strongly.

Recall some well-known facts from the theory of linear operators. First, for any closed operator $T$ on a Banach space, one has
\[
 \|(T-z)^{-1}\|\geq\frac{1}{\dist(z,\sigma(T))}, \quad \forall z\in\rho(T).
\]
On the other hand, recall that for a bounded operator $B$ it holds
\[
 \|(B-z)^{-1}\|\leq\frac{1}{\dist(z,D_{\|B\|})}, \quad \forall z,\; |z|>\|B\|,
\]
where $D_{\|B\|}=\{z\in\bC \mid |z|\leq\|B\|\}$. Let us remark that $\sigma(B)\subset D_{\|B\|}$. 

Let $\{a_{n}\}_{n\geq1}, \{b_{n}\}_{n\geq0}\subset\bC$ and $a_{n}\neq0$, $\forall n\geq1$. Further $J_{n}$ denotes the $(n+1)\times(n+1)$
Jacobi matrix of the form
\[    
     J_{n}=\begin{pmatrix}
	b_0 & a_1\\
	a_1 & b_1 & a_2\\
	& a_2 & b_2 & a_3\\
	& & \ddots & \ddots & \ddots\\
	& & & a_{n-1} & b_{n-1} & a_n\\
	& & & & a_n & b_n
      \end{pmatrix}.
\]
Recall also that the polynomial sequence $\{p_n\}_{n\geq0}$ determined by recurrence
\[
 a_{n}p_{n-1}(x)+b_{n}p_{n}(x)+a_{n+1}p_{n+1}(x)=xp_{n}(x), \quad n\geq 0,
\]
and initial conditions $p_{-1}=0$ and $p_{0}=1$, is related with $J_{n}$ by formula
\[
 p_{n}(x)=\left(\prod_{k=1}^{n}\frac{1}{a_{k}}\right)\det\left(x-J_{n-1}\right), \quad n\geq1.
\]

\begin{lemma}{\cite[Lem. 2.2]{KvA}}\label{lem:KVlem}
 Let $\|J_{n}\|\leq M$, then
 \[
  \left|\frac{p_{n}(z)}{a_{n+1}p_{n+1}(z)}\right|\leq\frac{1}{\dist(z,D_{M})}=\frac{1}{|z|-M}, \quad \forall z\in\bC, |z|>M.
 \]
 On the other hand, one has
 \[
  \left|\frac{p_{n}(z)}{a_{n+1}p_{n+1}(z)}\right|\geq\frac{1}{2|z|} \quad \forall z\in\bC, |z|>3M.
 \]

\end{lemma}

\begin{proof}
 By simple linear algebra, 
 \[
  \frac{p_{n}(z)}{a_{n+1}p_{n+1}(z)}=\frac{\det\left(z-J_{n-1}\right)}{\det\left(z-J_{n}\right)}=\langle e_{n}, (z-J_n)^{-1} e_{n} \rangle
 \]
 (note that $e_{n}$ stands for the $(n+1)$-th vector of the standard basis of $\bC^{n+1}$).
 Consequently,
 \[
 \left|\frac{p_{n}(z)}{a_{n+1}p_{n+1}(z)}\right|=|\langle e_{n}, (z-J_n)^{-1} e_{n} \rangle|\leq\|(z-J_n)^{-1}\|\leq\frac{1}{\dist(z,D_{\|J_{n}\|})}\leq\frac{1}{\dist(z,D_{M})},
 \]
 whenever $|z|>M$.
 
 One the other hand,
 \[
  |\langle e_{n}, (z-J_n)^{-1} e_{n} \rangle|=\frac{1}{|z|}\left|1+\langle e_n,\sum_{k=1}^{\infty}z^{-k}J_{n}^{k}e_{n}\rangle\right|
  \geq\frac{1}{|z|}\left(1-\frac{\|J_{n}\|}{|z|-\|J_{n}\|}\right)\geq\frac{1}{2|z|},
 \]
 whenever $|z|\geq3M$.
\end{proof}

Observe that if $R_{n}\in\bC^{n+1,n+1}$ is the permutation matrix determined by equations $R_{n}e_{k}=e_{n-k}$, $\forall k\in\{0,1,\dots,n\}$, then
\[
 R_{n}J_{n}R_{n}=\begin{pmatrix}
	b_n & a_n\\
	a_n & b_{n-1} & a_{n-1}\\
	& a_{n-1} & b_{n-2} & a_{n-2}\\
	& & \ddots & \ddots & \ddots\\
	& & & a_{2} & b_{1} & a_1\\
	& & & & a_1 & b_0
      \end{pmatrix},
\]
$\det(J_{n}-z)=\det(R_{n}J_{n}R_{n}-z)$, and
\[
 \langle e_{n}, (J_{n}-z)^{-1} e_{n} \rangle = \langle e_{0}, (R_{n}J_{n}R_{n}-z)^{-1} e_{0} \rangle, \quad \forall z\in\rho(J_{n}).
\]
If convenient, we identify matrix $R_{n}$ with the operator $R_{n}\oplus 0$ acting on $\ell^{2}(\NN)$

In what follows, $\{a_{n,N} \mid n,N\in\mathbb{N}\}\subset\bC$, $\{b_{n,N} \mid n\in\NN,\; N\in\mathbb{N}\}\subset\bC$, and 
\[
 J(N)=\begin{pmatrix}
	b_{0,N} & a_{1,N}\\
	a_{1,N} & b_{1,N} & a_{2,N}\\
	& a_{2,N} & b_{2,N} & a_{3,N}\\
	& & \ddots & \ddots & \ddots
      \end{pmatrix}
\]
stands for the semi-infinite (complex) Jacobi matrix. Finally, $P_{n}\in\mathcal{B}(\ell^{2}(\NN))$ denotes the orthogonal projection on 
$\spanl\{e_{0},\dots,e_{n}\}$.

\begin{proposition}\label{prop:strong_conv}
 Let
 \begin{equation}
  \lim_{n/N\to t}a_{n,N}=A \quad \mbox{ and } \lim_{n/N\to t}b_{n,N}=B.
 \label{eq:conv_assum}
 \end{equation}
 Further assume that there exists $\epsilon>0$ such that
 \begin{equation}
  \sup\{\|P_{n}J(N)P_{n}\| \mid |n/N-t|<\epsilon\}<\infty.
 \label{eq:unif_bound_cond}
 \end{equation}
 Then
 \[
  \slim_{n/N\to t} R_{n}J(N)R_{n}=J(A,B)
 \]
 where $J(A,B)\in\mathcal{B}(\ell^{2}(\NN))$ stands for the Jacobi matrix with constant diagonal $B$ and constant off-diagonal $A$.
\end{proposition}

\begin{proof}
 Recall an easily verifiable statement: Let $\mathcal{X}$ be a Banach space, $B\in\mathcal{B}(\mathcal{X})$ and $B_{n}$ a uniformly bounded sequence of operators acting on $\mathcal{X}$.
 If $B_{n}\varphi\to B\varphi$, as $n\to\infty$, for all $\varphi$ from a dense subset of $\mathcal{X}$, then $\slim_{\!n\to\infty} B_n=B$.
 
 Take arbitrary $\{n_{j}\},\{N_{j}\}\subset\N$, $N_{j}\to\infty$ and $n_{j}/N_{j}\to t$, for $j\to\infty$. Let us denote temporarily $\hat{J}_{j}=R_{n_{j}}J(N_{j})R_{n_j}$
  Assumption~\eqref{eq:unif_bound_cond} guarantees the existence of $j_{0}\in\N$ such that
 \[
  \sup_{j\geq j_{0}}\|P_{n_{j}}J(N_{j})P_{n_{j}}\|=\sup_{j\geq j_{0}}\|\hat{J}_{j}\|<\infty.
 \]
 Hence, taking into account the above statement, it suffices to verify that
 \[
  \lim_{j\to\infty}\hat{J}_{j}e_{n}=J(A,B)e_{n}, \quad \forall n\in\NN.
 \]
 
 Take $n\in\NN$ and $j>n$, then
 \[
  \|\hat{J}_{j}e_{n}-J(A,B)e_{n}\|^{2}=|a_{n_j-n+1,N_j}-A|^{2}+|b_{n_j-n,N_j}-B|^{2}+|a_{n_j-n,N_j}-A|^{2}\to0,
 \]
 for $j\to\infty$, by assumption \eqref{eq:conv_assum}. If $n=0$, set $a_{n_j+1,N_j}=0$ in the above equation.
 All in all, the claim is verified.
 \end{proof}

 It is again quite easy to see that for any $B_{n},B\in\mathcal{B}(\mathcal{X})$ such that $\sup_{n}\|B_{n}\|\leq M$ and $\slim_{\!n\to\infty}B_{n}=B$, one has
 \[
  \slim_{n\to\infty}(B_{n}-z)^{-1}=(B-z)^{-1}, \quad \forall z\notin D_{M}
 \]
 and the convergence is local uniform in $z$. Indeed, since
 \[
  (B-z)^{-1}-(B_{n}-z)^{-1}=(B_{n}-z)^{-1}(B_{n}-B)(B-z)^{-1}
 \]
 one obtains
 \[
 \|(B-z)^{-1}\varphi-(B_{n}-z)^{-1}\varphi\|\leq\frac{1}{|z|-M}\|(B_{n}-B)(B-z)^{-1}\varphi\|,
 \]
 from which the strong convergence of resolvents follows. The local uniformness follows, for example, from the Mantel's theorem.
 
 The next statement is in fact a corollary of Proposition \ref{prop:strong_conv}, however, it is also a complex generalization of \cite[Thm.~2.1]{KvA}.
 Therefore we formulate it as a proposition.

 \begin{proposition}{\cite[Thm.~2.1]{KvA}}\label{prop:KV21}
  Let the assumptions \eqref{eq:conv_assum} and \eqref{eq:unif_bound_cond} hold. Denote the value of the supremum in \eqref{eq:unif_bound_cond} by $M$. Then
  \[
   \lim_{n/N\to t}\langle e_n, \left(z-P_{n}J(N)P_{n}\right)^{-1}e_{n}\rangle=\frac{2}{z-B+\sqrt{(z-B)^{2}-4A^{2}}}
  \]
  locally uniformly in $\bC\setminus D_{M}$.
 \end{proposition}
 
 \begin{remark}
  Note that by \eqref{eq:conv_assum},
  \[
   M=\sup\{\|P_{n}J(N)P_{n}\| \mid |n/N-t|<\epsilon\}\geq |B|+2|A|.
  \]
 \end{remark}

 \begin{proof}
  It follows from Proposition \ref{prop:strong_conv} that
  \[
   \slim_{n/N\to t}\left(z-R_{n}J(N)R_{n}\right)^{-1}=\left(z-J(A,B)\right)^{-1}
  \]
  locally uniformly in $z\notin D_{M}$. Hence,
  \begin{align*}
   \lim_{n/N\to t}\langle e_n, \left(z-P_{n}J(N)P_{n}\right)^{-1}e_{n}\rangle&=\lim_{n/N\to t}\langle e_0, \left(z-R_{n}J(N)R_{n}\right)^{-1}e_{0}\rangle\\
   &=\langle e_0, \left(z-J(A,B)\right)^{-1}e_{0}\rangle
  \end{align*}
  locally uniformly in $z\notin D_{M}$. It is a standard result that 
  \[
  \langle e_0, \left(z-J(A,B)\right)^{-1}e_{0}\rangle=\frac{2}{z-B+\sqrt{(z-B)^{2}-4A^{2}}},
  \]
  for all $z\notin[B-2A,B+2A]$ (a line segment in $\bC$). To verify that one can show that
  \[
   \slim_{n\to\infty}P_{n}J(A,B)P_{n}=J(A,B)
  \]
  together with the formula  
  \[
   \langle e_0, \left(z-J_{n}(A,B)\right)^{-1}e_{0}\rangle=\frac{U_{n}\left(\frac{z-B}{2A}\right)}{A\ U_{n+1}\left(\frac{z-B}{2A}\right)}, \quad \forall z\notin[B-2A,B+2A]
  \]
  where $J_{n}(A,B)$ stands for the $(n+1)\times(n+1)$ truncation of $J(A,B)$, i.e., $J(A,B)=J_{n}(A,B)\oplus0$, and
  \[
   U_{n}(x)=\frac{\left(x+\sqrt{x^{2}-1}\right)^{n+1}-\left(x-\sqrt{x^{2}-1}\right)^{n+1}}{2\sqrt{x^{2}-1}}
  \]
  are Chebyshev polynomials of the second kind.
 \end{proof}

 Recall that the main result of paper \cite{KvA} in Theorem 1.4. which claims the following.
 
 \begin{theorem}{\cite[Thm.~1.4]{KvA}}\label{thm:KV14}
  If $\{a_{n,N} \mid n,N\in\mathbb{N}\}\subset\R_{+}$, $\{b_{n,N} \mid n\in\NN,\; N\in\mathbb{N}\}\subset\R$ and $\{p_{n,N} \mid n\in\NN,\; N\in\mathbb{N}\}$ associated family of orthonormal
  polynomials. Further let non-negative $a\in C(\R_{+})$ and real $b\in C(\R_{+})$ be given, such that
  \[
    \lim_{n/N\to t}a_{n,N}=a(t) \quad \mbox{ and } \lim_{n/N\to t}b_{n,N}=b(t),
  \]
  for all $t>0$. Then for the family of $\{\nu_{n,N} \mid n\in\NN,\; N\in\mathbb{N}\}$ of root-counting measures of polynomials $\{p_{n,N} \mid n\in\NN,\; N\in\mathbb{N}\}$, one has
  \[
   \wlim_{n/N\to t}\nu_{n,N}=\frac{1}{t}\int_{0}^{t}\omega_{[b(s)-2a(s),b(s)+2a(s)]}\dd s
  \]
  where $\omega_{[\alpha,\beta]}$ is absolutely continuous measure supported on $[\alpha,\beta]$ with density
  \[
   \frac{\dd \omega_{[\alpha,\beta]}}{\dd t}=\frac{1}{\pi\sqrt{(\beta-t)(t-\alpha)}},
  \]
  if $\alpha<\beta$. If $\alpha=\beta$, $\omega_{[\alpha,\beta]}=\delta_{\{\alpha\}}$.
 \end{theorem}

 For $t>0$, let us denote
 \[
  \sigma(t)=\frac{1}{t}\int_{0}^{t}\omega_{[b(s)-2a(s),b(s)+2a(s)]}\dd s.
 \]
 In the proof of Theorem \ref{thm:KV14}, authors prove that
 \begin{equation}
  \lim_{n/N\to t}U^{\nu_{n,N}}(z)=U^{\sigma(t)}(z)
  \label{eq:lim_logpot}
 \end{equation}
 locally uniformly in certain neighborhood of complex $\infty$ (i.e., for $|z|>M$), where $U^{\mu}$ denotes the logarithmic potential of the (compactly supported) Borel measure $\mu$.
 Under the assumptions of Theorem \ref{thm:KV14}, supports of all measures $\nu_{n,N}$, for all $n,N$ such that $n/N$ is close to $t$, are included in a real interval $[-M,M]$. This implies
 that the limit relation \eqref{eq:lim_logpot} holds true for all $z\notin[-M,M]$ and hence for almost all $z\in\bC$ (w.r.t. the Lebesgue measure). Under these conditions one can show (following 
 standard methods of Potential Theory - Widom's lemma) the weak convergence 
 \begin{equation}
   \wlim_{n/N\to t}\nu_{n,N}=\sigma(t).
 \label{eq:lim_weak}
 \end{equation}
 
 However, in the general case of complex sequences $a_{n,N}$ and $b_{n,N}$, one can get only the relation \eqref{eq:lim_logpot} outside a ball, $|z|>M$, and \textbf{not} for almost all $z\in\bC$.
 This however does not imply the weak convergence \eqref{eq:lim_weak}. To our best knowledge nor the existence of the weak limit is guaranteed (only a subsequence, by Helly's theorem). Thus, one can 
 not expect the validity of Theorem \ref{thm:KV14} in the complex setting. However, one can get at least the following.
 
\begin{proposition}\label{prop:complex}
 Let $a\in C([0,\infty))$ and $b\in C([0,\infty))$ be complex-valued functions and 
  \[
    \lim_{n/N\to t}a_{n,N}=a(t) \quad \mbox{ and } \lim_{n/N\to t}b_{n,N}=b(t),
  \]
  for all $t>0$. Then if the weak limit of root-counting measures
  \[
   \nu=\wlim_{n/N\to t}\nu_{n,N}
  \]
  exists, then $\nu$ and $\sigma(t)$ are equipotential measures, i.e., their logarithmic potentials coincide outside the union of their supports.
\end{proposition}

\begin{remark}
 Note the functions $a$ and $b$ are assumed to be continuous in $0$ (from the right). This additional condition simplifies the proof considerably and we do not aim here to achieve a full generality.
\end{remark}

\begin{proof}
 Note that coefficients $a_{n,N}$ and $b_{n,N}$ are uniformly bounded if $n/N$ is restricted to a compact subsets of $[0,\infty)$, as it follows from the assumptions.
 Take $t>0$ and $0<\epsilon<t$, then
 \[
  \sup\{|b_{n,N}| \mid |n/N-t|\leq\epsilon\}+2\sup\{|a_{n,N}| \mid |n/N-t|\leq\epsilon\}<\infty.
 \]
 Denote by $J_{n}(N)$ the $(n+1)\times(n+1)$ truncation of $J(N)$. Consequently, condition \eqref{eq:unif_bound_cond} is fulfilled and let us denote the uniform bound of operators $J_{n}(N)$, for
 $|n/N-t|\leq\epsilon$, by $M$.
 
 The following part proceeds analogously as the proof of \cite[Thm.~1.4]{KvA}. Since
 \[
  \frac{\det\left(z-J_{n-1}(N)\right)}{\det\left(z-J_{n}(N)\right)}=\langle e_{n}, (z-J_{n}(N))^{-1} e_{n} \rangle,
 \]
 one has 
 \[
  \det\left(z-J_{n}(N)\right)=\prod_{k=0}^{n}\frac{1}{\langle e_{k}, (z-J_{k}(N))^{-1} e_{k}\rangle}, \quad |z|>M.
 \]
 Thus,
 \[
  U^{\nu_{n,N}}(z)=\frac{1}{n+1}\log\left|\det\left(z-J_{n}(N)\right)\right|=-\frac{1}{n+1}\sum_{k=0}^{n}\log\left|\langle e_{k}, (z-J_{k}(N))^{-1} e_{k}\rangle\right|,
 \]
 or equivalently
 \[
  U^{\nu_{n,N}}(z)=-\int_{0}^{1}\log\left|\langle e_{[ns]}, (z-J_{[ns]}(N))^{-1} e_{[ns]}\rangle\right|\dd s, \quad |z|>M.
 \]
 As $n/N\to t$, one has $[sn]/N\to st$. Hence, by Proposition \ref{prop:KV21},
 \[
  \lim_{n/N\to t}\langle e_{[ns]}, (z-J_{[ns]}(N))^{-1} e_{[ns]}\rangle=\frac{2}{z-b(st)+\sqrt{(z-b(st))^{2}-4a(st)^{2}}}.
 \]
 Further, by Lemma \ref{lem:KVlem}, one has
 \[
  \frac{1}{2|z|}\leq\left|\langle e_{[ns]}, (z-J_{[ns]}(N))^{-1} e_{[ns]}\rangle\right|\leq\frac{1}{|z|-M},
 \]
 for $|z|>3M$. Consequently, the Lebesgue's dominated convergence theorem applies and we get
 \begin{align*}
  \lim_{n/N\to t}U^{\nu_{n,N}}(z)&=\int_{0}^{1}\log\left| \frac{z-b(st)}{2}+\sqrt{\left(\frac{z-b(st)}{2}\right)^{2}-a(st)^{2}}\right| \dd s\\
				 &=\frac{1}{t}\int_{0}^{t}\log\left| \frac{z-b(s)}{2}+\sqrt{\left(\frac{z-b(s)}{2}\right)^{2}-a(s)^{2}}\right| \dd s\\
 \end{align*}
 for $|z|>3M$. The function in the last integral is known to coincide with the logarithmic potential of $\omega_{[b(s)-2a(s),b(s)+2a(s)]}$ at $z$.
 All in all, we obtained
 \[
 U^{\nu}(z)=\frac{1}{t}\int_{0}^{t}U^{\omega_{[b(s)-2a(s),b(s)+2a(s)]}}(z)\dd s=U^{\sigma(t)}(z),
 \]
 for $|z|>3M$. By the harmonicity of logarithmic potentials $U^{\mu}$ outsides the support $\mu$ and Identity principle for harmonic functions the last equality
 can be extended to all $z\notin(\supp\nu \cup\supp\sigma(t))$.
\end{proof}

\end{document}